\documentclass[11pt,twoside]{article}
\usepackage[soda]{optional} 
\usepackage{fullpage}

\usepackage{amsmath,amssymb,cite}
\usepackage{multirow,xspace}
\usepackage{microtype}
\usepackage{ordinalref}
\usepackage[sanseriff, basic]{complexity}
\usepackage[title]{appendix}
\usepackage{comment}

\newcommand{\spara}[1]{\smallskip\noindent{\bf #1}}

\newcommand{\squishlist}{\begin{list}{$\bullet$}
  { \setlength{\itemsep}{0pt}
     \setlength{\parsep}{3pt}
     \setlength{\topsep}{3pt}
     \setlength{\partopsep}{0pt}
     \setlength{\leftmargin}{1.5em}
     \setlength{\labelwidth}{1em}
     \setlength{\labelsep}{0.5em} } }
\newcommand{\squishend}{
  \end{list}  }

\newcommand{\eps}{\varepsilon}

\newcommand{\expt}[1]{2^{\poly(#1)}}
\newcommand{\szem}{{Szemer{\'e}di}}
\newcommand{\balls}{\code{QuickCluster}}
\DeclareMathOperator{\cost}{cost}
\DeclareMathOperator{\ideal}{ideal}
\newcommand{\opti}{\func{OPT}}
\DeclareMathOperator{\cl}{c\ell}
\DeclareMathOperator{\similarity}{sim}

\usepackage[shortlabels]{enumitem}

\newtheorem{theorem}{Theorem}[section]

\newtheorem{lemma}[theorem]{Lemma}
\newtheorem{definition}[theorem]{Definition}
\newtheorem{corollary}[theorem]{Corollary}

\newtheorem{remark}{Remark}[section]

\newcommand{\code}[1]{{\textproc{#1}}}
\usepackage{algorithm}
\usepackage[noend]{algpseudocode}
\makeatletter
\renewcommand{\ALG@beginalgorithmic}{\small}
\makeatother

\opt{soda}{
\renewcommand{\Statex}{}
}
\opt{full}{
}

\newcommand{\naturals}{{\mathbb N}}

\newcommand{\abs}[1]{\lvert{#1}\rvert}
\newcommand{\norm}[1]{\lVert{#1}\rVert}

\providecommand{\poly}{{\operatorname{poly}}}

\newenvironment{prooftext}[1]{\par\noindent{\bf Proof#1.}\quad}{\nopagebreak$\qed$\\}
\newenvironment{proof}{\begin{prooftext}{}}{\nopagebreak\end{prooftext}}
\newenvironment{proofof}[1]{\begin{prooftext}{ of #1}}{\end{prooftext}}
\newcommand{\qed}{\hfill\hbox{${\vcenter{\vbox{ \hrule height 0.4pt\hbox{\vrule width 0.4pt height 6pt \kern5pt\vrule width 0.4pt}\hrule height 0.4pt}}}$}}

\newcommand{\setnosp}[1]{\{{#1}\}}
\makeatletter
\newcommand{\set}{\@ifstar
                     \setnosp%
                     \setnosp%
}
\makeatother

\newcommand{\mycomment}[1]{}

\newcommand{\rst}[1]{{\ensuremath{{\mathbin\upharpoonright}\raise-.9ex\hbox{${\scriptstyle{#1}}$}}}}
\newcommand{\subst}[1]{{\ensuremath{\raise-.9ex\hbox{${\scriptstyle{#1}}$}}}}

\DeclareMathOperator*{\expect}{\mathbb{E}}

\newcommand{\indic}{\mathbb{I}\,}

\renewcommand{\epsilon}{\varepsilon}

\DeclareRobustCommand{\calA}[0]{{\mathcal A}}
\DeclareRobustCommand{\calB}[0]{{\mathcal B}}
\DeclareRobustCommand{\calC}[0]{{\mathcal C}}

\DeclareRobustCommand{\calG}[0]{{\mathcal G}}

\DeclareRobustCommand{\calR}[0]{{\mathcal R}}
\DeclareRobustCommand{\calS}[0]{{\mathcal S}}

\mycomment{
  \newcommand{\proceedings}[1]{Proceedings of the \ordinal{#1}}
  \newcommand{\stoc}[1]{\proceedings{#1} ACM Symposium on Theory of Computing (STOC)}
  \newcommand{\focs}[1]{\proceedings{#1} IEEE Symposium on Foundations of Computer Science (FOCS)}
  
  \newcommand{\soda}[1]{\proceedings{#1} ACM-SIAM Symposium on Discrete Algorithms (SODA)}

  \newcommand{\icalp}[1]{\proceedings{#1} International Colloquium on Automata, Languages and Programming (ICALP)}
  \newcommand{\stacs}[1]{\proceedings{#1} International Symposium on Theoretical Aspects of Computer Science (STACS)}
  
  \newcommand{\ics}[1]{\proceedings{#1} Symposium on Innovations in Theoretical Computer Science (ITCS)}
  \newcommand{\colt}[1]{\proceedings{#1} Conference on Learning Theory (COLT)}

  \newcommand{\kdd}[1]{\proceedings{#1}
ACM SIGKDD Conference on Knowledge Discovery and Data Mining (KDD)}

}
  \newcommand{\proceedings}[1]{Proc. of \ordinal{#1}}
  \newcommand{\stoc}[1]{\proceedings{#1} STOC}
  \newcommand{\focs}[1]{\proceedings{#1} FOCS}
  
  \newcommand{\soda}[1]{\proceedings{#1} SODA}

  \newcommand{\icalp}[1]{\proceedings{#1} ICALP}
  \newcommand{\stacs}[1]{\proceedings{#1} STACS}
  
  \newcommand{\ics}[1]{\proceedings{#1} ITCS}
  \newcommand{\colt}[1]{\proceedings{#1} COLT}
  
  \newcommand{\kdd}[1]{\proceedings{#1} KDD}
\title{Local correlation clustering}

\author{
\begin{tabular}{ccc}
 Francesco Bonchi &
\hspace{4mm} David Garc\'ia--Soriano &
Konstantin Kutzkov\\
\multicolumn{2}{c}{Yahoo Labs} & IT University of Copenhagen\\
\multicolumn{2}{c}{Barcelona, Spain} &  Copenhagen, Denmark\\
\multicolumn{2}{c}{\small{\sf \{bonchi, davidgs\}@yahoo-inc.com}} & {\small{\sf konk@itu.dk}}
\end{tabular}
}
\date{}

\begin{document}
\maketitle
\sloppy

\begin{abstract}
\emph{Correlation clustering} is perhaps the most natural formulation of clustering. Given $n$ objects and a pairwise similarity measure, the goal is
to cluster the objects so that, to the best possible extent, similar objects are put in the same cluster and dissimilar objects are put in different clusters.

Despite its theoretical appeal, the practical relevance of correlation clustering still remains largely unexplored. This is mainly due to the fact that correlation clustering requires the $\Theta(n^2)$ pairwise similarities as input. In large datasets this is infeasible to compute or even only to store.

In this paper we initiate the investigation into \emph{local} algorithms  for correlation clustering, laying the theoretical foundations for clustering ``big data''.  In \emph{local correlation clustering} we are
given the identifier of a single object and we want to return the cluster to which it belongs in some globally consistent near-optimal clustering, using a small number of similarity queries.

Local algorithms for correlation clustering open the door to \emph{sublinear-time} algorithms, which are particularly useful when the similarity between items is costly to compute, as it is often the case in many practical application domains.
They also imply $(i)$ distributed and streaming clustering algorithms,  $(ii)$ constant-time estimators and testers for cluster edit
distance, and $(iii)$ property-preserving parallel reconstruction algorithms for clusterability.

Specifically, we devise a local clustering algorithm attaining a $(3, \eps)$-approximation (a solution with cost at most $3\cdot \opti + \eps n^2$, where $\opti$ is the optimal cost). Its running time is $O(1/\eps^2)$ independently of the dataset size. If desired, an explicit approximate clustering for all $n$ objects can be produced in time
$O(n/\eps)$ (which is provably optimal).
We also provide a fully additive $(1,\eps)$-approximation with local query
complexity $\poly(1/\eps)$ and time complexity $\expt{1/\eps}$. The explicit clustering can be found in time $ n\cdot \poly(1/\eps)
+ 2^{\poly(1/\eps)}$. 
The latter yields the fastest polynomial-time approximation scheme for
correlation clustering known to date.

\end{abstract}

\newpage

\section{Introduction}\label{sec:intro}
In \emph{correlation clustering}\footnote{\scriptsize Sometimes called \emph{clustering with qualitative information} or \emph{cluster editing}.} we are given a set $V$ of $n$ objects and a pairwise similarity function~$\similarity: V \times V \rightarrow [0,1]$,
and  the goal is to cluster the items in such a way that, to the best possible extent, similar objects are put in the same cluster and dissimilar objects are put in different clusters. Assuming that cluster identifiers are represented by natural numbers,
a clustering $\cl$ is a function $\cl:V \rightarrow \mathbb{N}$. Correlation clustering aims at
minimizing the cost:
\begin{equation}
\label{equation:correlation-clustering} 	
\mspace{-20.0mu} \sum_{\substack{(x,y) \in V \times V, \\ \cl(x)=\cl(y)}} (1-\similarity(x,y)) \mspace{5.0mu} +
\mspace{-15.0mu} \sum_{\substack{(x,y) \in V \times V, \\ \cl(x)\not=\cl(y)}} \similarity(x,y).
\end{equation}
The intuition underlying the above problem definition is that
if two objects $x$ and $y$ are assigned to the same cluster we
should pay the amount of their dissimilarity $(1-\similarity(x,y))$, while
if they are assigned to different clusters we should pay the
amount of their similarity $\similarity(x,y)$.

In the most widely studied setting, the similarity function is binary, i.e.,  $\similarity: V \times V \rightarrow \{0,1\}$.
This setting can be viewed very conveniently trough graph-theoretic lenses: the $n$ items correspond to the vertices of a \emph{similarity graph} $G$, which is a complete undirected graph with edges labelled ``+'' or ``-''.
An edge $e$ causes a \emph{disagreement} (of \emph{cost} 1) between the similarity graph and a
clustering when it is a ``+'' edge connecting vertices in different clusters, or a ``--''
edge connecting vertices within the same cluster.
If we were given a \emph{cluster graph} \cite{cluster_editing} (or \emph{clusterable} graph), i.e., a graph whose set of positive edges is the union
of vertex-disjoint cliques, we would be able to produce a perfect (i.e., cost 0) clustering simply by computing the connected components of the positive graph.
However, similarities will generally be inconsistent with one another, so incurring a certain cost is unavoidable.
Correlation clustering aims at minimizing such cost.
The problem can be viewed as an agnostic learning problem, where we try to approximate the adjacency function of $G$ by the hypothesis class of
cluster graphs; alternatively, it is the task of finding the equivalence relation that most closely resembles a given symmetric relation $R$.

Correlation clustering provides a general framework in which one only needs to define a suitable similarity function.
This makes it particularly appealing for the task of clustering structured objects, where the similarity function is domain-specific and does not rely on an ad hoc specification of some suitable metric such as the Euclidean distance of vectors. Thanks to this generality, the technique is applicable to a multitude of problems in different domains, including duplicate detection and  similarity joins~\cite{duplicate_detection,corr_weighted}, biology~\cite{clustering_genes}, image segmentation~\cite{image_segmentation} and social networks~\cite{chromatic_clustering}.

Another key feature of correlation clustering is that it does not require a prefixed number of clusters, instead it automatically finds the optimal number.

Despite its appeal, correlation clustering has been, so far, mainly of theoretical interest. This is due to its  scaling behavior with the size of the
input data: given $n$ items to be clustered, building the complete similarity graph $G$ requires $\Theta(n^2)$ similarity computations. For a large $n$, the the similarity graph $G$ might unfeasible to construct, or even only to store. This is the main bottleneck of correlation clustering and the reason why
its  practical relevance still remains largely unexplored.

The high-level contribution of our work is to overcome the main drawback of correlation clustering, making it scalable. We achieve this by designing algorithms that can construct a clustering in a \emph{local} and distributed manner.

The input of a local  clustering algorithm is the identifier of one of the $n$ objects to be clustered, along with a short random seed. After making a small number of oracle similarity queries (probes into the pairwise similarity matrix), a local algorithm outputs the cluster
to which the object belongs, in some globally consistent near-optimal clustering.

\subsection{A model for local correlation clustering}
In the following we focus on the binary case: we will discuss the non-binary case together with other extensions in Section \ref{sec:extensions}.
We work with the adjacency matrix model, which assumes oracle access to the input graph $G$. Namely, given $x, y \in V(G)$, we can ask whether $\{x, y\}$ is a positive edge of $G$; each query is charged with unit cost.
By \emph{explicitly finding} a clustering we mean storing $\cl(v)$ for every $v \in V(G)$.
In this explicit model a running time of $\Omega(n)$ is necessary as it requires to specify all values.
An algorithm with complexity $O(n)$ for (approximate) correlation clustering is already a significant improvement over the complexity of most current solutions, but we take a step further and ask whether the
dependence on $n$ may be avoided altogether by producing \emph{implicit} representations of the cluster mapping.

It is for this reason that we define local clustering as follows.
Let us fix, for each finite graph $G$, a collection $\calC^G$ of ``high quality'' clusterings for $G$.
\begin{definition}[Local clustering algorithm]\label{def:local}
Let $t \in \naturals$. A clustering algorithm $\calA$ for $\calG^n$ is said to be \emph{local} with time (resp., query) complexity~$t$  if having
oracle access to any graph $G$, and taking as input $|V(G)|$ and a vertex $v \in V(G)$, $\calA$ returns  a cluster label $\calA^G(v)$ in
time $t$ (resp., with $t$ queries).

Algorithm $\calA$ implicitly defines a clustering, described by the cluster label function $\cl(v) =
\calA^G(v)$, where the {same} sequence $r$ of random bits is used by $\calA$ to calculate $\calA^G(v)$ for each~$v$.
The success probability of~$\calA$ is the infimum (over all graphs $G$) of the probability (over $r$) that the clustering implicitly defined by $\calA^G$ belongs to $\calC^G$.
\end{definition}
Note that $t$ does not depend on $n = |V(G)|$: this means that the cluster label of each vertex can be computed
in constant time independently of the others. On the other hand, $t$ could have a (hopefully mild) dependence on the desired \emph{quality} of the
clustering produced  (which defines the set~$\calC^G$ for a given~$G$), and the success probability of~$\calA$.
Finally, it is important to note that, in order to define a unique ``global'' clustering across different
vertices, the same sequence $r$ of random coin flips must be used.

Sometimes we also allow \emph{local algorithms with preprocessing $p$}, meaning (when $p$ denotes time complexity) that $\calA$ is
allowed to perform computations and queries using total time $p$ before reading the input vertex~$v$. This preprocessing computation/query set is
common to all vertices and may only depend on the outcome of $\calA$'s internal coin tosses and the edges probed.

\subsection{Contributions and practical implications}
We focus on approximation algorithms for local correlation clustering with sublinear time and query complexity.
Since any multiplicative approximation needs to make $\Omega(n^2)$ queries (Section~\ref{sec:lb}),
we need less stringent requirements.\footnote{\scriptsize
We remark that in a different model that uses \emph{neighborhood oracles}~\cite{correlation_revisited}, it is possible to bypass the $\Omega(n^2)$ lower bound for multiplicative approximations that holds for edge queries. In fact from our analysis we can derive the first sublinear-time constant-factor approximation algorithm for this case; see Section~\ref{sec:extensions}.
}
One way is to allow an additional $\eps$-fraction of edges to be violated, compared to the optimal clustering
of cost $\opti$. Following Parnas and Ron~\cite{approx_vc}, we study $(c, \eps)$ approximations: solutions with at most $c\cdot \opti
+ \eps\cdot n^2$ disagreements. These solutions form the set $C^G$ of ``high-quality'' clusterings in Definition~\ref{def:local}. Here~$c$ is a small
constant and $\eps \in (0,1)$ is an accuracy parameter specified by the user. Essentially $\eps$ handles the trade-off between the desired accuracy and the run-time:  the larger $\eps$ the faster then algorithm, but also the further from $\opti$.

While we provide the formal statement of our results in Section \ref{sec:results}, here we highlight
the main message of this paper: there exist efficient local clustering algorithms with good approximation guarantees. Namely, in time $t = \poly(1/\eps)$
it is possible to obtain $(O(1),\eps)$-approximations locally. (Typically we think of $\eps$ as a user-defined constant.)
This yields many practical contributions as by-products:
\squishlist
    \item \textbf{Explicit clustering in time $\boldmath{O(n)}$.} Given that $\cl(v)$ can be computed in time $t$ for each $v \in
    V(G)$, one  can produce an explicit clustering in time  $O(n \cdot t)$. Since $t = O(1)$, this is
    linear in the number of vertices (not edges) of the graph.
    More generally, the complexity of finding clusters of a subset $S\subseteq V$ of vertices requested by the user  is proportional to the size of this subset.

    \item \textbf{Distributed algorithms.} We can assign vertices to different processors and compute their cluster labels in parallel, provided that the same random seed is passed along to all
         processors. 

    \item \textbf{Streaming algorithms.} Similarly, local clustering algorithms can cluster graphs in the streaming setting, where edges arrive in arbitrary order. In this case the sublinear behaviour is lost because we still need to process every edge. However, the memory footprint of the algorithm can be brought down from $\Omega(n^2)$ to
    $O(n)$ (called the \emph{semi-streaming} model~\cite{semi_stream}).
    Indeed, note that given a fixed random seed,  for every vertex $v$ the set of all possible queries  $Q_v$ that can be made  during the computation of $\cl(v)$   has size\footnote{\scriptsize This bound  can in fact be reduced to $t$ for the non-adaptive algorithms we devise.} at most $2^t$.
    This set can be computed before any edge arrives. From then on it suffices to keep  in memory the edges $(v, w)$ where $w \in Q_v$, and there are $n \cdot 2^t =     O(n)$ of them.
In fact, the running time of the local-based algorithm will be dominated by the time it takes to discard the unneeded edges.

    \item \textbf{Cluster edit distance estimators and testers.} We can estimate the degree of clusterability of the input data in  constant time by sampling pairs of
    vertices and using  the local clustering algorithm to see
    how many of them disagree with the input graph.  We believe this can be an important primitive to develop new algorithms.
    Moreover, estimators for cluster edit distance give (tolerant) testers for the property of being clusterable, thereby allowing us to quickly detect data instances where any attempt to obtain a good clustering is bound to failure.

    \item \textbf{Local clustering reconstruction.} Queries of the form \emph{``are $x, y$ in the same cluster?''} can be answered in constant time without having to partition the whole graph: simply compute $\cl(x)$ and $\cl(y)$, and check for equality.
          This means that we can ``correct'' our input graph $G$ (a ``corrupted'' version of a clusterable graph) so that the
          modified graph we output is close to the input and satisfies the property of being clusterable.
This fits the paradigm of \emph{local property-preserving data reconstruction} of
\cite{data_reconstr} and~\cite{mono_reconstr}.
\squishend

To the best of our knowledge, this is the first work about local algorithms for correlation clustering.

\section{Background and related work}
\spara{Correlation clustering.}
Minimizing disagreements is the same as maximizing agreements for exact algorithms, but the two tasks differ with regard to
approximation. Following~\cite{fixed_clusters}, we refer to these two problems as $\code{MaxAgree}$ and $\code{MinDisagree}$, while  $\code{MaxAgree}[k]$ and $\code{MinDisagree}[k]$ refer to the variants of the problem with a bound $k$ on the number of clusters.
Not surprisingly $\code{MaxAgree}$ and $\code{MinDisagree}$
are $\NP$-complete~\cite{correlation_clustering, cluster_editing}; the same holds for their bounded
counterparts, provided that $k \ge 2$. Therefore approximate solutions are of interest. For $\code{MaxAgree}$, there is
a (randomized) \PTAS:
the first such result was due to Bansal \emph{et al.}~\cite{correlation_clustering} and ran in time $n^2 \exp{(O(1/\eps))}$, later
improved to $n \cdot 2^{\poly(1/\eps)}$ by Giotis and Guruswami~\cite{fixed_clusters}. The latter also presented a $\PTAS$ for $\code{MaxAgree}[k]$ that runs in time $n \cdot k^{O(\eps^{-3} \log(k/\eps))}$.
In contrast, $\code{MinDisagree}$ is $\APX$-hard~\cite{cluster_qualitative}, so we do not expect a
\PTAS. Nevertheless, there are constant-factor approximation algorithms~\cite{correlation_clustering,cluster_qualitative,balls}.
The best factor ($2.5$) was given by Ailon \emph{et al.}~\cite{balls}, who also present a simple, elegant algorithm that achieves
a slightly weaker expected approximation ratio of $3$, called $\balls$ (see
Section~\ref{sec:main}). 
For $\code{MinDisagree}[k]$, $\PTAS$ appeared in ~\cite{fixed_clusters} and~\cite{lt_gb}.
There is also work on correlation clustering on incomplete
graphs~\cite{correlation_clustering,cluster_qualitative,clustering_sdp,fixed_clusters,corr_weighted}.

%

\spara{Sublinear clustering algorithms.}
Sublinear clustering algorithms for geometric data sets are known~\cite{testing_clustering,sublinear_approx_cluster,clustering_similarity,sublinear_clustering,sublinear_clustering2}.
Many of these find implicit representations of the clustering they output.  There is a natural implicit
representation for most of this problems, e.g., the set of $k$ cluster centers. 
By contrast, in correlation clustering there may be no clear way to define a clustering for the whole graph based on a small set of
vertices. The only sublinear-time algorithm known for correlation clustering is the aforementioned result of ~\cite{fixed_clusters};
it runs in time~$O(n)$, but the multiplicative constant hidden in the notation has an exponential dependence on the approximation parameter.

The literature on \emph{active clustering} also contains algorithms with sublinear query complexity (see, e.g., \cite{active_clustering}); many of
them are heuristic or do not apply to correlation clustering. Ailon \emph{et al.}~\cite{active_queries} obtain algorithms for $\code{MinDisagree}[k]$ with sublinear query complexity, but the running time of their solutions is exponential in $n$.

\spara{Local algorithms.}
The following notion of locality is used in the distributed computing literature. Each vertex of a sparse graph is assigned a processor, and
each processor can compute a certain function in a constant number of rounds by passing messages to its neighbours (see Suomela's survey~\cite{survey_local}). Our algorithms are also local in this sense.

 Recently, Rubinfeld \emph{et al.}~\cite{fast_local} introduced a model that encompasses notions from several
algorithmic subfields, such as locally decodable codes, local reconstruction and local distributed computation.
Our definition fits into their framework:  it corresponds to \emph{query-oblivious, parallelizable, strongly local} algorithms that compute a cluster label function in constant time.

Finally, we point out the work of Spielman and Teng~\cite{spielman_local} pertaining local clustering algorithms. In their papers an algorithm is
``local'' if it can, given a vertex
$v$, output $v$'s cluster in time nearly linear in the cluster's size. Our
local clustering algorithms also have this ability (assuming, as they do, that for each vertex we are given a list of its neighbours), although
the results are not comparable because
~\cite{spielman_local} attempt to minimize the cluster's conductance.

\spara{Testing and estimating clusterability.}
Our methods can also be used for quickly testing clusterability of a given input graph $G$, which is related to the task of estimating the \emph{cluster edit distance},
    i.e., the minimum number of edge label swaps (from ``+'' to ``--'' and viceversa) needed to transform $G$ into a cluster graph. Note that this corresponds to the optimal cost of correlation clustering for the given input $G$.
Clusterability is a hereditary graph property (closed under
removal and renaming of vertices), hence it can be tested with one-sided error using a constant number of queries by the powerful result of Alon
and Shapira~\cite{charac_graphs}.
Combined with the work of Fischer and
Newman~\cite{testing_estimation}, this also yields estimators for cluster edit distance that run in time independent of the graph size. 
Unfortunately, the query complexity of the algorithm given by these results would be a tower exponential of height $\poly(1/\eps)$, where
$\eps$ is the approximation parameter.

Approximation algorithms for \code{MIN-2-CSP} problems~\cite{approx_csp} also give estimators for cluster edit distance.  However, they provide no way
of computing each variable assignment in constant time. Moreover, they use time $\sim n^2$ to calculate all assignments, and hence do not lend
themselves to sublinear-time clustering algorithms.

\section{Statement of results}\label{sec:results}
All our graphs are undirected and simple. For a vertex $v$, $\Gamma^+(v)$ is the set
of positive edges incident with $v$; similarly define $\Gamma^-(v)$. We extend this notation to sets of vertices in the obvious manner.
The  \emph{distance} between two graphs $G=(V,E)$ and $G'=(V,E')$ is  $|E \oplus E'|$.
Their fractional distance is their distance divided by $n^2$ (note this is in the interval $[0, 1/2)$).
Two graphs are \emph{$\eps$-close} to each other if their distance is at
most most $\eps n^2$.
A \emph{$k$-clusterable} graph is a union of at most $k$ vertex-disjoint cliques.  A graph is \emph{clusterable} if it is
$k$-clusterable for some $k$.

The following folklore lemma says that approximate $k$-clustering algorithms yield approximate clustering
algorithms with an unbounded number of clusters:
\begin{lemma}\label{fromkton}
If $G$ is clusterable, then it is $\eps$-close to $(1 + 1/\eps)$-clusterable.
\end{lemma}
\begin{proof}
Take the optimal clustering for $G$. Let $B$ be the set of vertices in clusters of size~$< \eps
n$. Now re-cluster the elements of $B$ arbitrarily into clusters of size $\lfloor \eps n \rfloor$ (except
possibly one). This introduces at most
$\eps n \cdot |B| \le \eps n^2$ additional errors. All but one of the clusters of the resulting
clustering have size $\ge \lfloor
\eps n \rfloor$, hence it has at most $1 + 1/\eps$ clusters.
\end{proof}

\begin{corollary}\label{cor:boundedk}
Any $(1,\eps/2)$ approximation to the optimal $(1+2/\eps)$-clustering is also
a $(1,\eps)$ approximation to the optimal clustering.
\end{corollary}
\begin{proof}
Immediate from the triangle inequality for graph distances.
\end{proof}

We are now ready to  summarize our results.
All our algorithms are (necessarily) randomized, and succeed with probability no less than $2/3$ (which can be amplified).
Our first result concerns the standard setting where the clusters of all vertices need to be
explicitly computed. We present a $(4,\eps)$-approximation\footnote{\scriptsize We can also produce an \emph{expected} $(3,\eps)$-approximation. Because we insist on algorithms that work with constant success probability, we talk about $(4,\eps)$-approximations, where the constant $4$ could be replaced with any number $\geq 3$.}
that runs in time $O(n/\eps)$; compare the $\Omega(n^2)$ complexity
of most other clustering methods. Our algorithm is optimal up to constant factors.
\begin{theorem}\label{explicit3}
Given $\eps \in (0, 1)$, a $(4,\eps)$-approximate clustering for $\code{MinDisagree}$ can be found in time $O(n/\eps)$.
Moreover, finding an $(O(1),\eps)$-approximation with constant success probability requires $\Omega(n/\eps)$ queries.
\end{theorem}
In other words, with a ``budget'' of $q$ queries we can obtain a $(4,O(n/q) )$-approximation.
In fact, the upper bound of Theorem~\ref{explicit3} can be derived from our next result.
It states that the same approximation
can be \emph{implicitly} constructed in constant time, regardless of the size of the graph.
\begin{theorem}\label{main_local}
Given $\eps \in (0, 1)$, a $(4,\eps)$-approximate clustering for $\code{MinDisagree}$ can be found locally in time $O(1/\eps^2)$,
  or in time $O(1/\eps)$ after preprocessing that uses $O(1/\eps^2)$ non-adaptive queries and time.
Moreover, finding an $(O(1),\eps)$-approximation with constant success probability requires $\Omega(1/\eps)$ adaptive queries.
\end{theorem}
As a corollary we obtain a partially tolerant tester of clusterability. We stress that the tester is
efficient both in terms of query complexity and time complexity, unlike many results in property testing.
\begin{corollary}\label{cor:prop_test}
There is a non-adaptive, two-sided error tester which accepts graphs that are $\eps / 5$-close to clusterable and rejects
graphs that are $\eps$-far from clusterable. It runs in time $O(1/\eps^2)$.
\end{corollary}
So far these results do not allow us to obtain clusterings that are arbitrarily close to the optimal one.
To overcome this issue, we also show (using different techniques) that a purely additive approximation can still be found
with $\poly(1/\eps)$ queries, but with an exponentially larger running time.

\begin{theorem}\label{main_dense}
Given $\eps \in (0,1)$, there is a local clustering algorithm that achieves an $(1,\eps)$ approximation to the cost of the optimal clustering.
Its local time complexity is $\poly(1/\eps)$ after preprocessing that uses $\poly(1/\eps)$ queries and $\expt{1/\eps}$ time.
\end{theorem}

For the explicit versions we obtain the following.

\begin{corollary}\label{ptas}
There is a $(1,\eps)$-approximate clustering algorithm for
$\code{MinDisagree}$ (and hence $\code{MaxAgree}$ too) that runs in time $n\cdot \poly(1/\eps) + 2^{\poly(1/\eps)}$.
In particular there is a $\PTAS$ for $\code{MaxAgree}$ with the same running time.
\end{corollary}

The ``in particular'' part follows from the observation that the optimum value for
$\code{MaxAgree}[2]$ is $\Omega(n^2)$ (see, e.g.,~\cite[Theorem 3.1]{fixed_clusters}).
The best $\PTAS$ in the literature~\cite{fixed_clusters} ran in time $n \cdot 2^{\Omega(\eps^{-3} \log^2 (1/\eps))}$.
In our result, the dominating term (depending on $n$) has an exponentially
smaller multiplicative constant (polynomial in $1/\eps$), and then we have an additive term exponential in $1/\eps$ (and independent of $n$).
As for lower bounds, observe that the $\Omega(n/\eps)$ bound from Theorem~\ref{explicit3} still applies,
   while the presence of a term of the form $2^{(1/\eps)^{\Omega(1)}}$ for very small $\eps$ seems hard to avoid
due to the $\NP$-completeness of the problems, as an optimal
solution can be found upon setting $\eps = 1/n^2$.

These results are established via the study of the corresponding problems with a prespecified number~$k$ of clusters; such
algorithms yield additive approximations to the general case upon setting $k = O(1/\eps)$\opt{full,soda}{ in view of Lemma~\ref{fromkton}}. For fixed $k$,
the bounds for our algorithms have the same form after replacing $\eps$ by $k/\eps$ (see Section~\ref{main_dense}). For example, we get a $\PTAS$ for
$\code{MaxAgree}[k]$ in time $n\cdot \poly(k/\eps) + 2^{\poly(k/\eps)}$.
\begin{corollary}
For any $0 < \eps_1 < \eps_2$, there is a non-adaptive, one-sided error tester which accepts graphs that are $\eps_1$-close to clusterable and rejects
graphs that are $\eps_2$-far from clusterable. It has query complexity $\poly(1/\eps)$ and runs in time $\expt{1/\eps}$, where $\eps =\eps_2-\eps_1$.
\end{corollary}

\spara{Techniques and roadmap.} 
Our first local algorithm (Theorem~\ref{main_local}) is inspired by the $\balls$ algorithm of Ailon \emph{et al.}~\cite{balls}, which
resembles the greedy procedure for finding maximal independent sets. The main idea to make a local version is to define the clusters ``in reverse''.
We find a small set $P$ of ``cluster centers'' or ``pivots'' by looking at a small induced subgraph, and then we show a simple rule to define
an extended clustering for the whole graph in terms of the adjacencies of each particular vertex with $P$. As it turns out,
such $P$ can be obtained by a procedure that finds a constant-sized ``almost-dominating'' set of  vertices that are within distance two of most other vertices in the graph, in such a way that we can combine the expected 3-approximation guarantee of~\cite{balls}
with an additive error term. The algorithm and its analysis are given in Section~\ref{sec:main}.

The second local algorithm (Theorem~\ref{main_dense}) borrows ideas from the PTAS for dense $\code{MaxCut}$ of Frieze and Kannan~\cite{lowrank_approx}  and uses low-rank approximations to the adjacency matrix of the graph.
(Interestingly, while such approximations have been known for a long time, their implications for correlation clustering have been overlooked.)
 Notably, implicit descriptions of these
approximations are locally computable in constant time (polynomial in the inverse of the approximation parameter). We show that in order to look for near-optimal clusterings, we can restrict the search
to clusterings that ``respect'' a sufficiently fine weakly regular partition
of the graph. Then we argue that this can be used to implicitly define a good approximate clustering: to cluster a given vertex, we first determine its
piece in a regular partition, and then we look at which cluster contains this piece in the best coarsening of the partition. The details are in Section~\ref{sec:additive}.

The lower bounds, proven in Section~\ref{sec:lb}, are applications of Yao's lemma~\cite{Yao77}. Broadly speaking, we give the candidate algorithm a
perfect clustering of most vertices of the graph into
$t = O(1/\eps)$ clusters of equal size, and for each of the remaining vertices a ``secret'' cluster is chosen at random among these $t$. The optimal
clustering of the resulting graph has fractional cost $\eps / c$ for some constant $c > 1$. We then ask the algorithm to find clusters for the
remaining vertices, and show that it must make $\Omega(n/\eps)$ adaptive queries if it is to output a clustering with fractional cost no larger than $\eps$.

Finally, in Section~\ref{sec:extensions} we discuss several extensions, including the case of non-binary similarity measure.
%

\section{{\large $(3,\eps)$}-approximations}\label{sec:main}
First we describe the \balls\ algorithm of Ailon et al.~\cite{balls}. It selects a random pivot, creates a cluster with it and its positive
neighborhood, removes the cluster, and iterates on the induced subgraph remaining. Essentially it finds a maximal independent set in the positive graph.
%
%
%
When the graph is clusterable, it makes no errors. In~\cite{balls}, the authors show that the expected cost of the clustering found is at most
three times the optimum.

Note that determining the positive neighborhood $\Gamma^+(v)$ of a pivot $v$ takes
$n - 1$ queries to the adjacency oracle. The algorithm's worst-case complexity is
$\Theta(n^2)$: consider the graph with no positive edges. In fact its time and query
complexity is $O(n
        c)$, where $c$ is the average number of clusters found. This  suggests
attempting to partition the data into a small number of
clusters to minimize query complexity.

\mycomment{ False! need to get edges outside given set of vertices.
    It is interesting to note that the set of \emph{positive} edges queried by \balls is the union of~$c$
    star graphs, each with the chosen pivot as center.  In particular the graph of positive edges
    queried is a forest with $c$ connected components; this
    implies that the number of edges queried that turn out to be positive is exactly $n - c$.
    Thus, given access to a more powerful neighbourhood oracle that, on input a vertex $v$ and an index $i$, told us the $i$th
    positive neighbor of $v$, $\balls$ would be much more efficient as it would run in $O(n)$.
}


We know from Lemma~\ref{fromkton} that any clustering can be $\eps$-approximated by a clustering with pieces of size $\Omega(\eps n)$.
So an idea would be to modify~$\balls$ so that most
clusters output are sufficiently large. Fortunately, $\balls$ tends to do just that on average, provided that the graph of positive edges is
sufficiently dense, because the expected size of the next cluster found is
precisely one plus the average degree of the remaining graph. Once the graph becomes too
        sparse, a low-cost clustering of the remaining vertices can be found without even looking at
        the edges, for example by putting each of them into their own cluster.\footnote{\scriptsize Another
            possibility that works is to cluster all remaining vertices into clusters of size $\eps
                n$, eliminating the
                need for singleton clusters.}

Another advantage of finding a small number of clusters is locality.
Let $P = P_1, \ldots, P_t$ denote the first $t$ elements of the sequence of pivots found by \balls.
Let us pick an arbitrary vertex $v$ contained in the neighbourhood of $\{P_1, \ldots, P_t\}$; all other
vertices can be safely ignored because as we shall see they usually will be incident to few edges (for suitably chosen $t$). Then the
pivot of $v$'s cluster is the \emph{first} element of $P$ that is a positive neighbour of $v$: therefore it can
be determined in time $O(t)$, assuming we are given the pivot sequence $P_1, \ldots, P_t$.

\begin{small}
\begin{algorithm}[t!]\label{alg:local}
\begin{algorithmic}[0]

\smallskip

\Function{LocalCluster}{$v,\eps$}
\State $P \gets \Call{FindGoodPivots}{\eps}$ \Comment{This is the preprocessing stage and can be taken outside}
\State \Return $\Call{FindCluster}{v, P}$
\EndFunction
\Statex
\smallskip

\Function{FindCluster}{$v, P$}
\If {$v \notin \Gamma^+(P)$} \Return $v$ \Comment{Cluster $v$ by itself}
\Else $\;\; i \gets \min \{j \mid v \in \Gamma^+(P_j)\}$; \Return $P_i$ \Comment{Find first positive neighbour in $P$}
\EndIf
\EndFunction
\Statex
\smallskip

\Function{FindGoodPivots}{$\eps$}
\For {$i \in [16]$}
    \State $P^i \gets \Call{FindPivots}{\eps/12}$;
    \State $\tilde{d}^i \gets$ estimate of the cost of $P^i$ with $O(1/\eps)$ local clustering calls (see Appendix~\ref{sec:main})
\EndFor
\State $j \gets \arg \min\{\tilde{d}^i \mid i \in [16]\}$
\State \Return $P^j$
\EndFunction
\Statex
\smallskip

\Function{FindPivots}{$\eps$}
\State Let $Q \subseteq V$ be a random sample without replacement of size $\min\big(n,\frac{1}{2\eps}\big)$
\State \Return $\Call{IndependentSet}{Q}$
\EndFunction
\Statex
\smallskip

\Function{IndependentSet}{$Q$}
\State $P \gets []$ (empty sequence)
\For {$v \in Q$}
    \If {\Call{FindCluster}{$v, P$} $= v$} append $v$ to $P$
    \EndIf \EndFor
\Return $P$
\EndFunction

\caption{\code{LocalCluster}}
\end{algorithmic}
\end{algorithm}
\end{small}

Therefore we propose the scheme whose pseudocode is given in Algorithm 1 (the analysis is presented in the next section).  Assuming we know a good sequence $P$, an implicit clustering is defined
deterministically in the way described above; two vertices $v$ and $v'$ belong to the same cluster if and
only if \Call{FindCluster}{$v$} = \Call{FindCluster}{$v'$}.
Similarly to \balls, we can find a set of pivots by finding an independent set of vertices in
the graph; to keep it small we restrict the search to an induced subgraph of size $O(1/\eps)$.
This is done by~\Call{FindPivots}{}, which can
be seen as a ``preprocessing stage'' for the local clustering algorithm \Call{FindCluster}{}.
In the next section the following key lemma will be shown.
\begin{lemma}\label{lem:eight}
Let $\epsilon \in (0,1)$ and $r, s > 1$. The expected cost of the clustering determined by
$\Call{FindPivots}{\eps}$ is at most $3 \cdot \opti + \epsilon n^2$, and the probability that it
exceeds $3r \cdot \opti + s \cdot \epsilon n^2$ is less than $\frac{1}{r} + \frac{1}{s}$.
\end{lemma}

For example, setting $r = 4/3$, $s = 8$ we see that with probability $1/8$, the clustering
determined by $\Call{FindPivots}{\eps / 4}$ is a $(4, \eps)$-approximation to the optimal one.
Although this low bound on the success probability may be overly pessimistic, we can amplify it in
order to obtain better theoretical guarantees. To do this with confidence $2/3$ we try several samples $Q$
and estimate the cost of the associated local clusterings by sampling random edges.



\begin{lemma}\label{lem:cost}
Let $d$ denote the fractional cost of the optimal clustering. With probability at least~$5/6$, $\Call{FindGoodPivots}{\eps}$ returns a pivot set $P^i$ with
fractional cost at most $4d + \eps$. 
Its running time is $O(1/\epsilon^2)$.
\end{lemma}
\opt{full}{
\begin{proof}
Denote by $d^i$ the fractional cost of the clustering defined by $P^i$. For any $\tau \in (0, 1)$, it holds that with probability $1-1/400$,
       $\Call{ClusteringError}{P^i, \tau}$ returns $\tilde{d^i}$ such that
$|\tilde{d^i}-d^i| \le \tau$. This is a standard application of the Chernoff bound: the error probability is at most
$ \exp(-2m \tau ^2) \le \exp(-6) \le 1/400. $
By the union bound, $|\tilde{d^i}-d^i| \le \eps / 3$ for all $i \in [16]$ except with probability $\le 1/25$.
By Lemma~\ref{lem:eight}, the probability that none of the clusterings defined $P^i$ is a $(4, \eps/3)$-approximation  is at most $(7/8)^{16}$.
Altogether we have that with probability at least $5/6$, both $d_i \le 4 d + \eps / 3$ and $\tilde{d}^i \le 4d + 2 \eps / 3$ hold for at least one $i$, whereas
for all $i$ with $d^i > 4d + \eps$ we have $\tilde{d}^i > 4d + 2 \eps / 3$. This means that the set
of pivots returned defines a clustering with fractional error at most $4d + \eps$.
\end{proof}
}

Finally, to obtain a purely local clustering algorithm with no preprocessing we need each vertex to
find the good set of pivots by itself, as per  \Call{LocalCluster}{$v,\eps$}.
            Note it is crucial here that all vertices have access to the same source of randomness so the same set of
            pivots is found on each separate call. In practical
            implementations this means introducing an additional parameter of short length, for
            instance a common random seed to be used.

From Lemmas~\ref{lem:eight} and~\ref{lem:cost} we can easily deduce two of our main results.

\begin{corollary}[Upper bound of Theorem~\ref{main_local}]$\Call{LocalCluster}{v, \epsilon}$ is a local clustering algorithm for $\code{MinDisagree}$ achieving a $(4,\eps)$ approximation to the optimal clustering with probability $2/3$.
The preprocessing runs in time $O(\min(n/\eps, 1/\eps^2))$, and the
clustering time per vertex is $O(1/\eps)$.
\end{corollary}

\begin{corollary}[Upper bound of Theorem~\ref{explicit3}]\label{explicit4}
An explicit clustering attaining a $(4, \eps)$ approximation can be found with probability $2/3$
in time $O(n/\eps)$.
\end{corollary}
By a very similar argument we can produce an \emph{expected} $(3,\eps)$-approximate clustering in time $O(n/\eps)$.
           \opt{full}{\\
\begin{proof}
Call \Call{FindGoodPivots}{$\eps$} once to obtain a good  pivot sequence $P$ with probability
$2/3$ in time $O(\min(n,\eps^{-2})) \le O(n)$. Then run $\Call{FindCluster}{v, P}$ sequentially for each vertex $v$ in order to determine
its cluster label $l$, appending $v$ to the list of vertices in cluster labelled $l$. Finally output
the resulting clusters. The whole process runs in time $O(n) + O(n/\epsilon) = O(n/\epsilon)$.
\end{proof}

           }
A time-efficient property tester for clusterability (Corollary~\ref{cor:prop_test}) is also a simple consequence of the above.

\opt{full}{
\begin{corollary}
There is a tester 
that accepts with probability at least $5/6$ if the graph is $\eps/5$-close to clusterable,
and rejects with probability at least $5/6$ if the graph is $\eps$-clusterable.
It runs in time $O(1/\eps^2)$.
\end{corollary}
\begin{proof}
The tester lets $P \gets \Call{FindGoodPivots}{\eps / 5}$ and accepts if
$\Call{ClusteringError}{P,\eps / 5}$ is at most $(1-1/15)\eps$.
\end{proof}
}

\subsection{Analysis of the local algorithm}
We prove the approximation guarantees of the algorithm (Lemma~\ref{lem:eight}) by comparing it to the clustering found by \balls, which is known to achieve an expected 3-approximation~\cite{correlation_clustering}.
In this section we consider the input graph $G=G^+$ with negative edges removed, so $\Gamma(v) = \Gamma^+(v)$
for all $v \in G$.

    The following is a straightforward consequence of the multiplicative Chernoff bounds:
    \begin{lemma}
    Let $c > 1$, $\epsilon, \delta \in (0, 1)$, $m \in \naturals^+$ and $n = \frac{3\cdot c}{(c-1)^2}\frac{\log(m/\delta)}{\eps}$.
    Suppose $\{X^i_j \mid i \in [m], j \in [n]\}$ are random variables in $[0,1]$ such that for all $i \in [m]$,
            the variables $X^i_1, \ldots, X^i_n$ are independent with common mean $\mu_i$.
     Define $\tilde{\mu}_i = \frac{1}{n} \sum_{j=1}^n X_i^j$. Then with probability at least $1-\delta$, the following holds:
    \begin{itemize}
        \item If $\min_{i \in [m]}\mu_i \le \frac{1}{c} \cdot \eps$, then $\min_{i\in[m]} \tilde{\mu}_i \le \eps$.
        \item For all $i \in[m]$, if $\mu_i > c \cdot \eps$, then $\tilde{\mu}_i > \eps$.
    \end{itemize}

    \begin{corollary}\label{cor:aggr}
    Let $d, \delta \in (0, 1), c > 1$.
    Let $C_1, \ldots, C_m$ be $m$ clusterings such that at least one of them has fractional cost $\le d$. Then with probability $1-\delta$ we can select
    $i \in [m]$ such that $C_i$ has fractional cost at most $c \cdot d$
    using a total of $\frac{3c}{(c-1)^2} \frac{m\, \log(m/\delta)}{\eps}$ edge queries to $C_1,\ldots,C_m$.
    \end{corollary}

    \end{lemma}

\begin{proofof}{Lemma~\ref{lem:cost}}
\mycomment{
    Denote by $d^i$ the fractional cost of the clustering defined by $P^i$. For any $\tau \in (0, 1)$, it holds that with probability $1-1/400$,
           $\Call{ClusteringError}{P^i, \tau}$ returns $\tilde{d^i}$ such that
    $|\tilde{d^i}-d^i| \le \tau$. This is a standard application of the Chernoff bound: the error probability is at most
    $ \exp(-2m \tau ^2) \le \exp(-6) \le 1/400. $
    By the union bound, $|\tilde{d^i}-d^i| \le \eps / 3$ for all $i \in [16]$ except with probability $\le 1/25$.
    By Lemma~\ref{lem:eight}, the probability that none of the clusterings defined $P^i$ is a $(4, \eps/3)$-approximation  is at most $(7/8)^{16}$.
    By Corollary~\ref{cor:aggr},
    Altogether we have that with probability at least $5/6$, both $d_i \le 4 d + \eps / 3$ and $\tilde{d}^i \le 4d + 2 \eps / 3$ hold for at least one $i$, whereas
    for all $i$ with $d^i > 4d + \eps$ we have $\tilde{d}^i > 4d + 2 \eps / 3$. This means that the set
    of pivots returned defines a clustering with fractional error at most $4d + \eps$.
}
Use Lemma~\ref{lem:eight} and Corollary~\ref{cor:aggr}.
\end{proofof}

\begin{proofof}{Corollary~\ref{explicit4}}
Call \Call{FindGoodPivots}{$\eps$} once to obtain a good  pivot sequence $P$ with probability
$2/3$ in time $O(\min(n,\eps^{-2})) \le O(n)$. Then run $\Call{FindCluster}{v, P}$ sequentially for each vertex $v$ in order to determine
its cluster label $l$, appending $v$ to the list of vertices in cluster labelled $l$. Finally output
the resulting clusters. The whole process runs in time $O(n) + O(n/\epsilon) = O(n/\epsilon)$.
\end{proofof}

A \emph{partial clustering} of $V$ is a clustering of a subset $W$ of $V$.
Its \emph{partial cost} is the number of disagreements between edges that have at least one endpoint in $W$.

Now consider a clustering $\calC$ of $V$ into $C_1, \ldots, C_m \subseteq V$. For $S \subseteq [m]$, the
$S$th \emph{partial subclustering} of
$\calC$ is the partition of $V_S = \bigcup_{i \in S} C_i$ into $\{C_i\}_{i \in S}$.
Clearly the cost of a clustering upper bounds the partial cost of any of its partial
subclusterings.

\begin{lemma}\label{lem:prefix}
Let $P_1, \ldots, P_t$ denote the sequence of pivots found
by $\Call{FindPivots}{\eps}$. The expected  number of edge violations involving vertices within distance
$\le 2$ from $P_1, \ldots, P_{t}$ is at most $3 \cdot \opti$.
\end{lemma}

\begin{proof}
To simplify the analysis, in the proof of this lemma we modify~$\balls$ and ~\code{FindPivots} slightly
so that they run deterministically provided that a random permutation~$\pi$ of the vertex set $V$
is chosen in advance. Concretely, we consider a deterministic version of \balls, denoted
$\balls^\pi$, that uses
pivot set $\Call{IndependentSet}{Q}$, where $Q$ lists all vertices of $V$ in ascending order of $\pi$.
Similarly, deterministic $\code{FindCluster}^\pi$ takes for $Q$ the set of the \emph{first}
$O(1/\epsilon)$ elements in increasing order of $\pi$.
Clearly running $\code{FindCluster}^{\pi}$ on a random permutation~$\pi$ is the same as
running the original~\code{FindCluster}, and likewise for~\balls.

Observe that the set $P_1,
        \ldots, P_{t(\pi)}$ of pivots returned by $\code{FindCluster}^\pi$ is a prefix of
        the set of pivots returned by $\balls^\pi$. Therefore the
first~$t(\pi)$ clusters are the same as well, i.e., $P_1, \ldots, P_{t(\pi)}$ define a partial
subclustering of the one found by $\balls^\pi$. Hence the partial cost of the
subclustering determined by $\code{FindCluster}^\pi$ is in expectation at most $3 \cdot \opti$. This is
equivalent to the statement of the lemma.
\end{proof}

Next we show that \code{FindCluster} returns a small ``almost-dominating'' set of
vertices, in the sense quantified in the following result.
\begin{theorem}\label{lem:indep}
Let $G = (V, E)$ be a graph and $Q$ be an ordered sample of $r$ independent vertices uniformly
chosen with replacement from~$V$.
Let $P = \Call{IndependentSet}{Q}$. Then the expected number of edges of $G$ not incident with an
element of $P \cup \Gamma(P)$
is less than $\frac{n^2}{2r}$.
\end{theorem}
Observe that an existential result for an almost-dominating set $P$ is easy to establish by picking pivots in order of decreasing
degree in the residual graph. However, doing so would invalidate the approximation guarantees of
$\balls$ we are relying on.
We defer the proof of Theorem~\ref{lem:indep}. Assuming the result, we are ready to prove Lemma~\ref{lem:eight}.

\smallskip

\begin{proofof}{Lemma~\ref{lem:eight}}
Lemma~\ref{lem:prefix} says that the set of pivots found define a partial clustering with expected
cost bounded by $3\cdot \opti$.
           Let $Q$ be the random sample used by $\Call{FindPivots}{\eps}$.
           Theorem~\ref{lem:indep} is stated for sampling with replacement, but this implies the
           same result for sampling without replacement, so its conclusion still holds.
           Combining the two results and setting $r=1/(2\eps)$, we see that the set of
disagreements in the clustering produced can be written as the union of two disjoint random sets $A, B
\in V \times V$ with $\expect[A] \le 3 \cdot \opti$ and $\expect[B] \le \epsilon n^2$. Thus the total
cost is $|A \cup B| = |A| + |B|$.
By linearity of expectation, the expected cost is  $\expect[|A| + |B|] \le 3 \cdot \opti + \epsilon
n^2,$ and by applying Markov's inequality to the non-negative variables $|A|$ and $|B|$ separately we conclude that
$$\Pr[(|A| > r \expect[|A|]) \text{ or } (|B| > s \expect[|B|]) ]
< \frac{1}{r} + \frac{1}{s} .$$
\end{proofof}
The rest of this section is
devoted to prove Theorem~\ref{lem:indep}.
For any non-empty graph~$G$ and pivot $v \in V(G)$, let $N_v(G)$ denote the subgraph of $G$
resulting from removing all edges incident to $\{v\} \cup \Gamma(v)$ (keeping all vertices).
Define a random sequence $G_0, G_1, \ldots$ of graphs by $G_0 = G$ and $G_{i+1} =
N_{v_{i+1}}(G_i)$, where $v_1, v_2, \ldots$ are chosen independently at random from $V(G_0)$ (note
        that sometimes $G_{i+1} = G_i$).

\begin{lemma}\label{lem:del_edges}
Let $G_i$ have average degree $\tilde{d}$. When going from $G_i$ to $G_{i+1}$, the number of edges
decreases in expectation by at least $\binom{\tilde{d}+1}{2}$, and the number of degree-0 vertices
increases in expectation by at least $\tilde{d} + 1$.
\end{lemma}
\opt{soda}{
\begin{proof}
Let $V = V(G_0)$, $E = E^+(G_i)$ and let $d_u = |\Gamma(u)|$ denote the positive degree of $u \in V$
on $G_i$.
The claim on the number of degree-0 vertices is easy, so 
we prove the claim on the number of edges. Consider an edge $\{u, v\} \in E$. It is deleted if the chosen pivot is an element of $\Gamma(u)
\cup \Gamma(v)$ (this contains $u$ and $v$); let $X_{uv}$ be the 0-1 random variable associated with this event. It occurs with probability
$$ \expect[X_{uv}] = \frac{|\Gamma(u) \cup \Gamma(v)|}{n} \ge \frac{1 + \max(d_u, d_v)}{n} \ge
\frac{1}{n} + \frac{d_u + d_v}{2n}. $$
Let $D = \sum_{u < v \mid {\{u,v\}} \in E} X_{uv}$ be the number of edges deleted.
By linearity of expectation, its average is
\begin{eqnarray*}
 \expect[D] &=& \sum_{\substack{u < v\\ \{u, v\} \in E}} \expect[X_{uv}]
    \ge \frac{1}{2}  \sum_{\substack{u, v\\ \{u, v\} \in E}}\left(\frac{1}{n} + \frac{d_u+d_v}{2n}\right)\\
   &=&  \frac{\tilde{d}}{2} + \frac{1}{4n} \sum_{\substack{u, v\\ \{u, v\} \in E}} (d_u + d_v).
\end{eqnarray*}
Now we compute
\begin{eqnarray*}
       \frac{1}{4n} \sum_{\substack{u, v\\ \{u, v\} \in E}} (d_u + d_v)
    &=& \frac{1}{2n} \sum_{\substack{u, v\\ \{u, v\} \in E}} d_u
    = \frac{1}{2n} \sum_{u} d_u^2 \\
    &=& \frac{1}{2} \expect_u [d_u^2]
    \ge \frac{1}{2} \left(\expect_u [d_u]\right)^2
  \ge \frac{1}{2} \tilde{d}^2 ,
\end{eqnarray*}
hence
$ \expect[D] \ge \frac{\tilde{d}}{2} + \frac{\tilde{d}^2}{2} = \binom{\tilde{d}+1}{2}. $
\end{proof}
}

Now let $\widetilde{V}(G) = \{v \in V(G) \mid \deg(v) > 0 \}$ and define the ``actual size'' $S(G)$ of a
graph by $ S(G) = 2 \cdot |E(G)| + |\widetilde{V}(G)| = \sum_{v \in \widetilde{V}(G)} (1 + \deg(v)). $
Let $n = |V(G_0)|$ and define $\alpha_i \in [0,1]$ by \opt{full}{
$$\alpha_i = \frac{s(G_i)}{n^2}.$$ 
}\opt{soda}{$\alpha_i = \frac{s(G_i)}{n^2}.$
}

\begin{lemma}\label{lem:alpha_i}
For all $i \ge 1$ the following inequalities hold:
\begin{align}
\expect[ \alpha_i \mid v_1,\ldots,v_{i-1}] &\le \alpha_{i-1} (1 -
        \alpha_{i-1})\label{eq:cond_alpha_i},\\
\expect[ \alpha_i ] &\le \expect[\alpha_{i-1}] (1 - \expect[\alpha_{i-1}]) \label{eq:alpha_i},\\
\expect[ \alpha_i ] &< \frac{1}{i+1} \label{eq:alpha_i2}.
\end{align}
\end{lemma}
\opt{soda}{
\begin{proof}
Inequality~\eqref{eq:cond_alpha_i} is a restatement of Lemma~\ref{lem:del_edges}.
Inequality~\eqref{eq:alpha_i} follows from Jensen's inequality: since
     $ \expect[\alpha_i] = \expect\big[ \expect[\alpha_i \mid v_1,\ldots,v_{i-1}] \big] \le \expect[
     \alpha_{i-1} (1 - \alpha_{i-1})] $
and the function $g$ mapping $x$ to $g(x) = x (1 - x)$ is concave, we have
$ \expect[\alpha_i] \le \expect[g(\alpha_{i-1})] \le g(\expect[\alpha_{i-1}]) = \expect[\alpha_{i-1}] (1 -
        \expect[\alpha_{i-1}]). $

Finally we prove $\expect[\alpha_i] < 1/(i+1)$ for all $i\ge1$. We know that
$\expect[\alpha_1] \le \max_{x \in [0,1]} g(x) = g\left(\frac{1}{2}\right) = \frac{1}{4},$
so the claim follows by induction on $i$ as $g$ is increasing
on $[0, 1/2]$ and
$ g\left(\frac{1}{i}\right) = \frac{1}{i} - \frac{1}{i^2} \ge \frac{1}{i} - \frac{1}{i (i + 1)} = \frac{1}{i +
    1} . $
\end{proof}
}

\begin{remark}
With a finer analysis (Appendix~\ref{sec:sharper}), Equation~\eqref{eq:alpha_i2} can be strengthened to
$ \expect[\alpha_i] \le \frac{1}{\frac{1}{\widehat{a}} + i + \Omega(\ln{(\widehat{a}\cdot i)}) } , $
where $\widehat{a} = \min(\alpha_0, 1 - \alpha_0)$. (This does not affect the asymptotics.)
\end{remark}

\begin{proofof}{Theorem~\ref{lem:indep}}
Note that after sampling $r$ vertices $v_1, \ldots, v_r$ with replacement from $G_0$, the subgraph
of $G_0$ resulting from removing all edges incident to $\bigcup_{i=1}^r \{v_i\} \cup \Gamma(v_i)$ is
    distributed according to $G_r$. Using Equation~\eqref{eq:alpha_i2}, we bound $\expect[|E(G_r)|] \le
    \frac{n^2}{2} \expect[\alpha_r] < \frac{n^2}{2r}.$
\end{proofof}

\section{Fully additive approximations}\label{sec:additive}
Here we study $(1,\eps)$-approximations. By Lemma~\ref{fromkton} and its corollary, it is enough to
consider $k$-clusterings for $k = O(1/\eps)$.

\subsection{The regularity lemma.} One of the cornerstone results in graph theory is the regularity lemma of \szem, which has found a myriad
applications in combinatorics, number theory and theoretical computer science~\cite{reg_appl}. It asserts that every graph $G$ can be approximated by a small collection of random bipartite graphs;
in fact from $G$ we can construct a small ``reduced'' weighted graph $\widetilde{G}$ of
constant size which inherits many properties of $G$.
If we select an approximation parameter $\epsilon$, it gives
us an equitable partition of the vertex set $V$ of $G$ into a constant number $m = m(\eps)$ of classes
$C_1,\ldots, C_m$ such that the following holds:
  for any two large enough $A, B\subseteq V$, the number of edges between $A$ and
  $B$ can be estimated by thinking of $G$ as a random graph where the probability of an edge
  between $v \in C_i$ and $w \in C_j$ is $d(C_i, C_j) = |E(C_i, C_j)|/(|C_i||C_j|)$. (The precise notion of approximation we need will be
          explained later.)
Moreover, it is possible to choose a minimum partition size $m_{\min}$; often $m_{\min}$ is chosen such that
``internal'' edges among
vertices from the same class are few enough to be ignored.

The original result was existential, but algorithms to construct a regular
partition  are known~\cite{regalg_alon,weak_regularity,regalg_hyper} which run
in time polynomial in $|V|$ (for constant $\eps$). This naturally suggests trying to use the partition classes in order to obtain an approximation of the optimal clustering. Nevertheless, to the best of our knowledge, the only prior attempts to exploit the regularity lemma for
clustering are the papers of Speroto and Pelillo~\cite{regularity_clustering0} and
S{\'a}rk{\"o}zy, Song, \szem\  and Trivedi~\cite{regularity_clustering}. They use the
constructive versions of the lemma to find the reduced graph~$\widetilde{G}$, and apply standard clustering algorithms to $\widetilde{G}$.
Since the partition size~$m$ required by the lemma is an $1/\eps^5$-level
iterated exponential of $m_{min}$ (and this kind of growth rate is necessary~\cite{gowers_lb}), they propose
heuristics to avoid this tower-exponential behaviour. However, the running time of their algorithms is at least~$n^\omega$, where
$\omega \in [2, 2.373)$ is the exponent for matrix multiplication. Moreover, no theoretical bounds are provided on the quality of the clustering
found by working with the reduced graph, even if no heuristics were applied.

To address these issues, we opt to use a weaker variant of the regularity lemma due to Frieze and
Kannan~\cite{approx_matrices,approx_matrices_conf}. It  has better quantitative parameters and
gives an implicit description of the partition, which opens the door for local clustering.

\subsection{Cut decompositions of matrices}
The idea of Frieze and Kannan is to take any real matrix $A$ with row set $\calR$ and column set $\calS$ with bounded
entries and approximate it by a low-rank matrix of a certain form. (The case of interest for us is when $A$ is the adjacency matrix of a graph.)
Let $m = |\calR|$ and
$n = |\calS|$. Given $R \subseteq \calR, S \subseteq \calS$ and a real density $d$, the \emph{cut matrix} $D =
CUT(R, S, d)$
is the rank-1 matrix defined by $D_{ij} = d$ if $(i, j) \in R \times S$, and $0$ otherwise. We identify $R$ and
$S$ with with column indicator vectors of length $m$ and $n$, respectively. Then we can write
$CUT(R, S, d) = d \cdot R S^t$. We define
    $ A(R, S) = R^t A S = \sum_{i \in R} \sum_{j \in S} A_{ij}. $
A \emph{cut decomposition} of a matrix $A$ \emph{with relative error} $\eps$ is a set of cut matrices $\{ D_i \}_{i \in [s]}$, where
$D_i = CUT(R_i, S_i, d_i)$, such that for all $R \subseteq \calR, S \subseteq \calS$,
    $$ \vspace{-1 mm} \Bigl| A(R, S) - \sum_{i \in [s]} D_i(R, S) \Bigr| \le \epsilon  \cdot \norm{R}_2 \norm{S}_2
    \cdot \sqrt{m n}. $$
    Such a decomposition is said to have \emph{width} $s$ and \emph{coefficient length} $\sqrt{d_1^2+\ldots d_s^2}$.
\begin{theorem}[Cut decompositions~{\cite[Th. 1]{approx_matrices}}]\label{thm:fk}
Suppose $A$ is a~$\calR \times \calS$ matrix with entries in $[-1,1]$ and $\eps,\delta\in (0, 1)$
are reals.
Then in time
$\poly(1/\eps)/\delta$
we can, with probability $1-\delta$, find
implicitly a
    cut decomposition of width $\poly(1/\eps)$, relative error at most $\eps$ and coefficient length at most 6.
\end{theorem}

\spara{Regarding the meaning of ``implicit''.} By \emph{implicitly finding} a cut decomposition
$B = \sum_{i\in[s]} CUT(R_i, S_i, d_i)$  of a matrix $A$ in time $t$, we mean that
for any given pair $(x,y)\in\calC\times\calR$,
    we can compute all of the following in time $t$ by making queries to $A$:
\squishlist
    \item the rational values $d_1, \ldots, d_s$;
    \item the indicator functions $\indic[x \in R_i]$ and $\indic[y \in S_i]$, for all $i \in [s]$;
    \item the value of the entry $B_{x,y}$.
\squishend
In Appendix~\ref{app:fk} we give a sketch of how Frieze and Kannan achieve this. We also observe that their algorithm is non-adaptive.

\spara{Specifying a maximum cut-set size.} Suppose we start with arbitrary equitable partitions of the row set and column set of~$A$ into $t$ pieces. We can then find cut decompositions of the $t^2$ submatrices induced by the partition, and combine them into a cut decomposition of the original matrix that satisfies $|S_i| \le m/t, |T_i| \le
n/ t$; the reader may verify that this preserves the bound on relative error. This process can only increase the query and time
complexities by an $O(t^2)$ factor (c.f.~\cite[Section 5.1]{approx_matrices}).

\spara{Application to adjacency matrices.} Suppose $A$ is the adjacency matrix of an unweighted graph $G=(V,E)$, and identify the sets $\calR$
and $\calS$ with~$V$.  Then
$|\calR| = |\calS| = |V| = n$. 
Let $E(R, S)$ denote the number of edges between $R \subseteq V$ and
$S \subseteq V$. Then $A(R, S) = E(R, S) + E(R \cap S, R \cap S)$ and the conclusion of Theorem~\ref{thm:fk} can be written as
\begin{center} for all $R \subseteq \calR, S\subseteq \calS$, \end{center}
\begin{align}\label{eq:weak_reg}
E(R, S) + & E(R \cap S,  R \cap S)  =  \nonumber \\
& \Big(\sum_{i \in [s]}  d_i \cdot \lvert R\cap R_i \rvert \cdot |S \cap S_i|\Big) \pm \epsilon n \sqrt{|R| |S|}.
\end{align}
The last term can be bounded by $\eps n^2$.
While the standard regularity lemma supplies a much stronger notion of approximation,
      this bound suffices for certain applications.

\spara{Weakly regular partitions.}
A \emph{weakly $\eps$-pseudo-regular partition} of $V$ is
a partition of $V$ into classes $V^1, \ldots, V^\ell$ such that for all
disjoint $R, S\subseteq V$,
$\bigl| E(R, S) - \sum_{i,j \in [\ell]}  d(V^i, V^j) \cdot \lvert R\cap V^i \rvert \cdot |S \cap V^j| \bigr| \le \epsilon n^2,$
where $d(V^i, V^j)= \frac{E(V^i, V^j)}{|V^i| |V^j|}$. If, in addition, the partition is equitable, it is said to be \emph{weakly $\eps$-regular}.

Given a cut decomposition of a graph with relative error $\eps$ and size $s$, we get an $2\eps$-weakly pseudo-regular partition of size
$\ell \le 2^{2s} $
by taking the classes of the Venn diagram of $R_1, S_1, \ldots, R_s, S_s$ with universe $V$.
So we can enforce the condition that the sets $R_1, S_1,
   R_2, S_2,\ldots$ partition the vertex set of $G$, at an exponential increase in the number of such sets.
Furthermore, any weakly $\eps$-pseudo-regular partition of size $\ell$ may be refined to obtain a weakly $3 \eps$-regular partition of slightly larger
size; see~\cite[Section 5.1]{approx_matrices}.

Often the weak regularity lemma is stated thusly in terms of {weakly regular} partitions, but the formulation of Theorem~\ref{thm:fk} is
stronger in that it allows us to estimate the number of edges between two
sets in time $\poly(1/\eps)$ provided that we know the sizes of their intersections with all $R_i,
S_i$, even though the weakly regular partition has size
$\ell = 2^{\poly(1/\eps)}$.

\subsection{Near-optimal clusterings and the local algorithm}
Intuitively, two vertices $v, w$ in the same class of a regular partition have roughly the same number of
connections with vertices outside. Hence for any given clustering of the remaining nodes,
            the cost of placing $v$ into any one of the clusters is roughly the same as the cost of
            placing~$w$ there, suggesting they belong together in an optimal clustering (if we can afford to ignore the cost due to internal edges in the regular
                    partition). In other words, a regular partition can be ``coarsened'' into a good clustering;
the best one can be found by
considering all possible combinations of assigning partition classes to clusters and estimating the cost of each resulting clustering.

            We can make this argument rigorous by using bounds derived from the weak regularity lemma to approximate the cost of
            the optimal clustering by a certain quadratic program. If we ignore the terms with a single variable squared, the optimum of this program
            does not change by much as long as the partition is sufficiently fine. Then one can argue that the modified program attains its optimum for an
            assignment of variables which can be interpreted as a clustering that puts everything from the same regular partition into the same cluster.

\begin{lemma}\label{lem:opt_cl}
Let $A$ be the adjacency matrix of a graph $G =(V,E)$ and $k \in \naturals$.
Let $\{CUT(R_i, S_i, d_i)\}_{i \in [s]}$ be a cut decomposition of $A$ with relative error
$\frac{\eps}{2 k}$ and with $|S_i|, |T_i| \le \frac{\eps n}{8k}$ for all $i\in[s]$.
Denote by $C^*$ the optimal $k$-clustering, and by $C$ the optimal
$k$-clustering into classes that belong to the $\sigma$-algebra generated by $\bigcup_{i\in[s]} \{R_i, S_i\}$ over $V$.
Then $ \cost(C) - \epsilon n^2 \le \cost(C^*) \le \cost(C) .$
\end{lemma}
\begin{proof}
We use Equation~\eqref{eq:weak_reg} to introduce an ``idealized'' cost function $\ideal$
satisfying the following for any clustering $X$:
\begin{enumerate}
    \item $\abs{\cost(X) - \ideal(X)} \le \frac{\eps n^2}{2}$; and
    \item $\ideal(C) \le \ideal(X) + \frac{\eps n^2}{2}$.
\end{enumerate}
Taken together, these two properties imply the result.

For each $k$-clustering $X$ into $X_1, \ldots, X_k$, define
\begin{align}\label{eq:ideal}
\ideal(X)
  = {-\frac{n}{2}}\,\, &+\,\, \sum_{j \in [k]}\,\,                         \Bigg[\sum_{i \in [s]} \left(\frac{1 - d_i}{2}\right) |X_j \cap R_i| |X_j \cap S_i| \Bigg] \notag \\
  &+ \sum_{\substack{j,j' \in[k]\\j \neq j'}} \Bigg[\sum_{i \in [s]}      d_i  |X_j \cap R_i| |X_{j'} \cap S_i| \Bigg].
\end{align}
For any $j, j' \in [k]$, $j \neq j'$, using Equation~\eqref{eq:weak_reg} it holds that
\begin{alignat*}{3}
E(X_j, X_{j'}) &=             \Big[\, \sum_{i \in [s]} d_i  |X_j \cap R_i| |X_{j'} \cap S_i|  &\Big] \pm& \frac{\eps n}{2k} \sqrt{|X_j| |X_{j'}|}. \\
\intertext{Similarly,}
 E(X_j, X_j)   &= \frac{1}{2} \Big[ \sum_{i \in [s]} d_i  |X_j \cap R_i| |X_j \cap S_i|     &\Big] \pm& \frac{\eps n}{2k} \sqrt{|X_j| |X_{j}|}.
\end{alignat*}
Therefore 
\begin{align*}
   \cost(X) &= \sum_{j\in[k]} \binom{|X_j|}{2} - E(X_j, X_j) + \sum_{\substack{j,j'\in[k]\\j \neq j'}} E(X_j, X_j') \\
           &= -\frac{n}{2} +  \sum_{j\in[k]} \frac{1}{2} |X_j|^2 - E(X_j, X_j) + \sum_{\substack{j,j'\in[k]\\j \neq j'}} E(X_j, X_j')\\
           &\le \ideal(X) + \frac{\eps n}{2k} \cdot \sum_{\substack{j,j'\in[k]}} \sqrt{|X_j| |X_j'|} \\
           &= \ideal(X) + \frac{\eps n}{2k} \cdot \Big(\sum_{\substack{j\in[k]}} \sqrt{|X_j|}\Big)^2 \\
           &\le \ideal(X) + \frac{\eps n}{2k} \cdot k \Big(\sum_{\substack{j\in[k]}} {|X_j|}\Big)\\
           &= \ideal(X) + \frac{\eps n^2}{2},
\end{align*}
where the last inequality is by Cauchy-Schwarz.

It remains to be shown that $\ideal(X) \ge \ideal(C) - \frac{\eps n^2}{2}$; in other words, that there is an almost-optimal $k$-clustering
under the $\ideal$ cost function whose pieces are unions of the pieces  $V^1, \ldots, V^\ell$ of the Venn diagram of $S_1, T_1, \ldots, S_s, T_s$.
To see this, write
$$
R_i = \bigcup_{t \mid V_t \subseteq R_i} V^t,\quad
S_i = \bigcup_{t' \mid V_{t'} \subseteq S_i} V^{t'}.
$$
Then $$
|X_j \cap R_i| = \sum_{t \mid V_t \subseteq R_i} |X_j \cap V^t|,\quad
|X_j \cap S_i| = \sum_{t' \mid V_{t'} \subseteq S_i} |X_j \cap V^{t'}|.
$$
Therefore $\ideal(X)+n/2$ is a quadratic form on the $k \ell$ intersection sizes $|X_j \cap V^t|$,
$j \in [k], t \in [\ell]$:
$$ \ideal(X) + \frac{n}{2} = \sum_{\substack{j,j'\in[k]\\i\in [s]\\V^t\subseteq R_i,V^{t'}\subseteq S_i}} \lambda_{j,j'}^i |X_j \cap V^t| |X_{j'} \cap
V^{t'}|, $$
where $\lambda_{j,j}^i = (1-d_i)/2$ and $\lambda_{j,j'}^i = d_i$ when $j\neq j'$.


Now remove from this expression the terms where $t=t'$. 
Among these, the terms where $j \neq j'$ evaluate to zero because $X_j$ and $X_{j'}$ are disjoint. Each of the terms where $t=t'$ and $j = j'$ has absolute value
at most $$|\lambda_{j,j}^i| |X \cap V^t|^2 \le 4 |X| |V^t| \le \frac{\eps n}{2k} |X_j|,$$
since $|\lambda_{j,j}^i| = |\frac{1-d_i}{2}| \le 4$ from the bound on the coefficient length, and $|V^t| \le \eps n/(8k)$.
Therefore the term removal changes the value of the $\ideal$ cost function by at most $\eps n^2/2$.

For $(t, j) \in [\ell] \times [k]$, let $\alpha_j^t = |X_j \cap V^t|$.
Let $$\kappa_{j,j'}^{t,t'} = \indic[t\neq t'] \cdot \sum_{\substack{i\in[k]\\V_t \subseteq R_i\\V_{t'}\subseteq S_i}} \lambda_{j,j'}^i \, \indic[V^t
\subseteq R_i \wedge V^{t'} \subseteq S_i].$$
Then we have seen that
$$ \ideal(X) + \frac{\eps n^2}{2} \ge {-\frac{n}{2}} +\sum_{\substack{t,t'\in[\ell]\\j,j'\in[k]}} \kappa_{j,j'}^{t,t'}\, \alpha_j^t \alpha_{j'}^{t'} ,$$
and $\kappa_{j,j'}^{t,t} = 0$.
Hence finding the optimal $k$-clustering under the idealized cost function can be reduced, up to an additive error of $\frac{\eps}{2} n^2$, to solving the
following integer quadratic program:
\begin{alignat}{2}\label{eq:qp}
\text{minimize  }\quad{}&       {-\frac{n}{2}} + \sum  \kappa_{j,j'}^{t,t'}\,\alpha_{j}^t \alpha_{j'}^{t'} \\
\text{subject to}\quad{}&           \sum_{j\in[k]} \alpha_j^t  =|V^t|, \quad& \forall t\in [\ell] \notag\\
                             &\alpha_j^t                 \ge 0, \quad& \forall t\in[\ell], j\in[k] \notag\\
                             &\alpha_j^t  \in \naturals.\notag
\end{alignat}
The reason is that any feasible solution for $\{\alpha_j^t\}$ gives a clustering by assigning $\alpha_1^t$ arbitrary elements of $V^t$ to the first cluster,
    another $\alpha_2^t$ elements  of $V^t$ to the second cluster, and so on.

    Because $\kappa_{j,j'}^{t,t} = 0$, there is an optimal solution to~\eqref{eq:qp} in which for all $t \in[\ell]$, exactly one~$\alpha^i_t$ is equal to $|V^t|$ and the rest are zero.
    Indeed, fix $\alpha^{t'}_j$ for all $t'\neq t$ and all $j$ in a solution (which corresponds to fixing a $k$-clustering of $V \setminus
        V^t$). Then the objective function becomes a linear combination of $\alpha^t_1, \ldots, \alpha^t_{k}$, plus a constant term. Therefore it is minimized by picking the
cluster $j \in [k]$ with the smallest coefficient and setting $\alpha^t_j=|V^t|$.
\end{proof}
We sketch now our second local algorithm.

\begin{proofof}{Theorem~\ref{main_dense}}
For any $k \in \naturals, \eps \in (0, 1)$, we show a local algorithm that achieves an $(1,\eps)$-approximation to the optimal $k$-clustering in time
$\poly(k/\eps)$, after a preprocessing stage that uses $\poly(k/\eps)$ queries and $2^{\poly(k/\eps)}$
    time. Theorem~\ref{main_dense} then follows by setting $k = O(1/\eps)$.

Fist compute a cut decomposition of $A$ that satisfies the conditions of Lemma~\ref{lem:opt_cl}. By Theorem~\ref{thm:fk},
 it can be computed implicitly in $\poly(k/\eps)$
time. Let $V^1, \ldots, V^\ell$ be the atoms of the $\sigma$-algebra, where $\ell = 2^{2s}$ and $s = \poly(k/\eps)$. Observe that they can also be
defined implicitly:
given $x \in V$ we can compute in $\poly(s)$ time a $2s$-bit label that determines the unique $V^t$ to which $x$ belongs,
namely the value of the $2 s$ indicator functions $\indic[x \in S_i], \indic[x \in T_i]$.

Next we proceed to the more expensive preprocessing part. Consider a clustering all of whose classes are unions of $V^1,
    \ldots, V^\ell$. Any such clustering is defined by a mapping  $g:[\ell] \to [k]$ that, for every $i \in [\ell]$, identifies the cluster to which all elements of $V^i$
belong.
 We can try all the $\ell^k = \expt{k/\eps}$ possibilities for $g$, and for each of them and estimate the cost of the associated clustering
 to within $\eps/(2k)$ with high enough success probability by sampling. (We omit the details.)
If we select the best of them, by Lemma~\ref{lem:opt_cl}, it will have cost within $\eps n^2$ of the optimal one.

Now we have a ``best'' mapping $g$ from $[\ell]$ to $[k]$ that, for every $i \in \ell$, tells us the cluster of the elements of $V^i$. Finally, note that for any $x \in V$, the appropriate $i\in [\ell]$ such that $x \in V^i$ can be determined in time $\poly(s) =
\poly(k/\eps)$, and then we can get a cluster label for $x$ in time $\poly(k/\eps)$ by computing $g(i)$.
\end{proofof}

$
\vspace{-8.5mm}
$ 

\section{Lower bounds}\label{sec:lb}
We show that our algorithm from Section~\ref{sec:main} is optimal up to constant factors by
proving a matching lower bound for obtaining $(O(1),\eps)$-approximations.
For simplicity we consider expected approximations at first; later we prove that combining upper and lower bounds for expected approximations leads to
lower bounds for finding bounded
approximations with high confidence.

\begin{theorem}\label{lb_neps}
Let $c \ge 1$, $\eps \in(1/n, 1/(100 c))$.
Finding an expected $(c, \eps)$-approximation to the best
clustering with probability $1/2$ requires $\frac{n}{4000 \eps c^2}$ queries to the similarity matrix.

In addition, any local clustering algorithm achieving this approximation has query complexity $\Omega(1/(\eps c^2))$.
(This remains true even if we allow preprocessing, as long as its running time is bounded by a function of $\eps$ and not $n$.)
\end{theorem}
\begin{proof}
The first part implies the second because any $q(\eps)$-query local clustering algorithm with preprocessing $p(\eps)$ can be turned
into an explicit $n \cdot q(\eps) + p(\eps)$-query clustering algorithm. Given a lower bound of $n \cdot l(\eps)$ on the complexity of finding approximate
$(c,\eps)$-clusterings for large enough $n$, we get $n \cdot q(\eps) + p(\eps) \ge n \cdot l(\eps)$ for all large enough $n$, which implies $q(\eps)
    \ge l(\eps)$ since $\lim_{n->\infty} p(\eps) / n = 0$.
So we prove the first claim.

By Yao's minimax principle, it is enough to produce a distribution $\calG$ over graphs with the following properties:
\begin{itemize}
    \item the expected cost of the optimal clustering of $G \sim \calG$ is $\expect[\opti(G)] \le \frac{\eps n^2}{c}; $
    \item for any \emph{deterministic} algorithm making at most $n/(4000 \eps c^2)$ queries, the expected cost (over $G$) of the clustering produced exceeds $2\eps
    n^2 \ge c \cdot \expect[\opti(G)] + \eps n^2$.
\end{itemize}

Let $\alpha = \frac{1}{4 c}$, $k = \frac{1}{32 c \eps}$ and $l = \frac{k^2 \eps n}{3} \ge \frac{n}{4000 c^2 \eps}$.
We can assume that $c$, $k$ and $\alpha n / k$ are integral (here we use the fact that $\eps > 1/n$).
Let $A = \{1, \ldots, (1 - \alpha) n\}$ and
$B  = \{(1-\alpha) n + 1, \ldots, n\}$.
Consider the following distribution $\cal G$ of graphs:
partition the vertices of $A$ into
exactly $k$ equal-sized clusters $C_1, \ldots, C_k$. The set of positive edges will be the union of the cliques
defined by $C_1, \ldots, C_k$, plus an edge joining each vertex $v \in B$ to
a randomly chosen element $r_v \in A$.
Define the natural clustering of a graph $G \in \calG$ by the classes $C_i' = C_i \cup \{ v \in B\mid r_v \in C_i
\}$ ($i \in [k]$). This clustering will have a few disagreements because of the negative edges between different
vertices $v, w\in B$ with $r_v = r_{w}$. The cost of the optimal clustering of $G$ is bounded by that of
the natural clustering $N$, hence
$$ \expect [\opti]  \le \expect[\cost(N)] =  \frac{\binom{\alpha n}{2}}{k} \le \frac{\alpha^2 n^2}{2k} = \frac{\eps}{c} n^2. $$

We have to show that any algorithm making $l$ queries to graphs drawn from~$\calG$
produces a clustering with expected cost larger than $2 \eps n^2$.
This inequality holds provided that the output clustering $C$ and the natural clustering $N$ are
at least $3 \eps$-far apart. Indeed, reasoning about expected distances, $N$ is $\eps/c$-close to $G$, therefore any clustering that is
$2\eps$-close to $G$ is also $2\eps+\eps/c\le 3\eps$-close to $N$ from the triangle
inequality.

Since all graphs in $\calG$ induce the same subgraphs on $A$ and $B$ separately, we can assume
without loss of generality that the algorithm queries only edges between $A$ and $B$. Let us
analyze the distance between the natural clustering and the clustering found by the algorithm.
For $v \in B$,
let $Q_v$ denote set of queries it makes from $v$ to $A$ and put $q_v = |Q_v|$. Clearly we can assume $q_v \le k - 1$.
The total number of queries made is $q = \sum_{v \in B} q_v$.

As $r_v$ is independent of all edges from $[n] - \{v\}$ to
$[n]-\{v\}$, conditioning on the responses to all queries not involving $v$ we still know that
the probability that all responses are negative is $\Pr[r_v \notin Q_v] = 1 - q_v / k$.
When this happens, the probability that $r_v$ coincides with  the algorithm's choice is
at most $\frac{1}{k - q_v}$.

All in all we have that the probability that the algorithm puts $v$ into the same cluster as~$r_v$
is bounded by $\frac{1}{k - q_v} + \frac{q_v}{k}$. Let us associate a 0-1 random variable $a_v$ with
this event and put $R = \sum_{v\in B} a_v$. Consequently,
$$ \expect [R] \le \sum_{v \in B} \left(\frac{1}{k - q_v} + \frac{q_v}{k}\right) = \frac{q}{k} + \sum_{v \in B}
\frac{q_v}{k-q_v}.$$
We will see below (Lemma~\ref{kmqi}) that the last term can be bounded by $2(m+q) / k$,
   where $m = |B| = \alpha n$. Therefore
$\expect [R] \le \frac{3q + 2m}{k}$.

Now note that any vertex with $a_v = 0$ introduces $2 (n - m) / k \ge n / k$ new differences with the natural
clustering. Thus the expected number of differences is at least
\begin{eqnarray*}
\left(m - \frac{3q + 2m}{k}\right) \frac{n}{k} &=& m\left(1-\frac{2}{k}\right)\frac{n}{k} - \frac{3q n}{k^2} \\
 &\geq& \alpha \frac{n^2}{2k} -
\frac{3q n}{k^2} > 4 \eps n^2 - \frac{3q n}{k^2} \ge 3 \eps n^2,
\end{eqnarray*}
because $q \le l$.

\end{proof}
    \begin{lemma}\label{kmqi}
    Let $q_1, \ldots, q_m \in [0, k-1]$ with $\sum_{i=1}^m q_i = q$. Then
    $$ \sum_{i=1}^m \frac{1}{k - q_i} \le \frac{2(m+q)}{k}.$$
    \end{lemma}

    \begin{proof}
    Let $\gamma = \frac{q}{m+q}$. Define the sets
    $$ A = \{ i \in [m] \mid q_i \ge \gamma k \}$$
    and
    $$ B = \{ i \in [m] \mid q_i < \gamma k \}.$$
    Observe that $|A| \le \frac{q}{\gamma k} = \frac{m + q}{k}$.
    Then
    $$ \sum_{i=1}^m \frac{1}{k - q_i} \le |A| + \frac{|B|}{(1-\gamma) k} \le |A| + \frac{m}{(1-\gamma) k}
    \le \frac{2(m+q)}{k}.$$
    \end{proof}

Finally, we argue that similar bounds hold for algorithms that obtain
good approximation with high success probability.
\begin{lemma}\label{conv_expected}
Suppose~$\calA$ finds a           $(c,\eps)$-approximate clustering with success probability $1/2$ using $q$ queries, and
~$\calB$ finds an expected $(c,r \cdot \eps)        $-approximate clustering using $q$ queries.
Then there is an algorithm $\calC$ that finds an expected $\big(c,2\eps+\log(2r) \cdot \exp(-2 \eps^2 q)\big)$-approximation using $2 q \cdot \log(2r)$ queries.
\end{lemma}
\begin{proof}
Algorithm~$\calC$ does the following:
\begin{enumerate}
    \item Let $t \gets \log r$.
    \item Run $t$ independent instantiations of~$\calA$ to find clusterings $C_\calA^1, \ldots, C_{\calA}^t$ with $q t$ queries.
    \item Run~$\calB$ independently to find an expected $(c,r \cdot \eps)$-approximate clustering $C$ with $q$ queries.
    \item Estimate the quality of these $t+1$ clusterings using $q$ random samples for each of them.
    \item Return the clustering with the smallest estimated error.
\end{enumerate}
The query complexity bound of $\calC$ is as stated. When one of the $t + 1$ clusterings found is $(c,\eps)$-approximate, the probability that we fail to
return a $(c,2\eps)$-approximation is at most $p = \exp(-2 \eps^2 q) \cdot (t + 1)$.  In this case we bound the error of the clustering output by $1$. So the contribution to the expected approximation due to this kind
of failure is at most $(0, p)$. We assume from now on that this is not the case.

The probability that none of $C_{\calA}^1,\ldots, C_{\calA}^t$ is a $(c,\eps)$-approximation is at most
$2^{-t} \le \frac{1}{r}$. In this case we output a $(c,2\eps)$-approximation.
On the other hand, with probability at most $\frac{1}{r}$, we output a clustering that in expectation is a $(c,r \cdot \eps)$-approximation.
Therefore, the output is an expected $(c,2\eps) + \frac{1}{r}(c,r) + (0, p)$-approximation.
\end{proof}

\begin{corollary}
Let $\eps > 10^7/n$. Finding a $(c,\eps)$-approximate clustering with confidence $1/2$ requires $q = \frac{n}{2 \cdot 10^6 \cdot c^2 \eps}$  queries.
\end{corollary}

\begin{proof} We may assume $c \ge 3$.
Take the algorithm  $\calB$ from Corollary~\ref{explicit4} and plug it into Lemma~\ref{conv_expected}. This gives an
expected approximation of  $(\max(c,3), 2\eps + 25 \exp(-2 \eps^2 q)) \le (c, 3\eps)$ using $50 q$ queries.
The result now follows from Lemma~\ref{lb_neps}.\end{proof}

\section{Extensions}\label{sec:extensions}

\spara{Non-binary similarity function.} In Section 1 we have introduced correlation clustering in its most general form, with a pairwise similarity function~$\similarity: V \times V \rightarrow [0,1]$, while the case we ave studied so far is that of a binary
similarity function~$\similarity: V \times V \rightarrow \{0,1\}$.
%
The general case can be reduced to this by ``rounding the graph'', i.e., by
replacing a non-binary similarity score with either $0$ or $1$ according to which is the closest (breaking ties arbitrarily): Bansal \emph{et al.}~{\cite[Thm. 23]{correlation_clustering}}
showed that if $\calA$ is an algorithm that produces a clustering on a graph $G$ with $0,1$-edges with
approximation ratio $\rho$, then running $\calA$ on the rounding of $G$ achieves an approximation
ratio of $2\rho+1$.
Therefore our algorithms also provide $(O(1),\eps)$ approximations for correlation clustering  in the more general weighted case.


\spara{Neighborhood oracles.} If, given $v$, we can obtain a linked list of the \emph{positive} neighbours of~$v$ (in time linear in its length), then it is possible to obtain a
multiplicative $(O(1), 0)$-approximation in time $O(n^{3/2})$, which is sublinear. Indeed, Ailon and Liberty~\cite{correlation_revisited} argue that with a neighborhood oracle, $\balls$ runs in time $O(n +
        OPT)$; if $OPT \le n^{3/2}$ this is $O(n^{3/2})$. On the other hand, if we set $\eps = n^{-1/2}$ in our algorithm, we obtain in time
$O(n^{3/2})$ a $(O(1), n^{-1/2})$-approximation, which is also a $(O(1), 0)$-approximation when $OPT \ge n^{3/2}$. So we can run $\balls$ for
$O(n^{3/2})$ steps  and output the result; if it doesn't finish, we run our algorithm with $\eps = n^{-1/2}$.

\spara{Distributed/streaming clustering.}
In Section~\ref{sec:intro} we mentioned that there are general transformations from local clustering algorithms into distributed/streaming
algorithms. For  our local algorithm from Section~\ref{sec:main} we can do the following. Suppose that to each processor $P$ is assigned a subset $A_P$ of the pairs $V \times V$, so that $P$ can compute (or has information about) whether there is a
positive edge between $x$ and $y$ for the pairs $(x, y) \in A_P$. (The assignment of vertex pairs to processors can be arbitrary, as long as they
        partition $V \times V$.) Then each processor selects the same random vertex subset $S \subseteq V$ of size $O(1/\eps)$, and discards (or does
            not query/compute) the edges not incident with $S$ among those it can see ($A_P$).
After this, each processor outputs, for each $v$, the
pairs $(v, w)$ in $A_P$ (note that for each $v$, there are only $O(1/\eps)$ different such pairs). With this information the pivot set $T$ is
the subset of $S$ with no neighbour smaller than itself (in some random order), and then the label of $v$'s cluster is the first element of $T$ adjacent to $v$. This can be computed easily in another round.

Note that the sum of the memory used by all processors is $O(n/\eps)$, so for constant $\eps$ we also get a (semi-)streaming algorithm that makes one pass over the data,
with edges arriving in arbitrary order. In two passes we can reduce memory usage to $O(n+1/\eps^2)$: first store the adjacency matrix of the
subgraph induced by the random sample $S$, and compute the set $T$ of pivots. In the second pass, keep an integer for each $v$ that indicates the
first element of $T$ that has appeared as a neighbour of $v$ in the edges seen so far. At the end, this integer will be $v$'s cluster.


\spara{Variants of correlation clustering.}
The second algorithm (based on cut matrices) can easily be extended to chromatic correlation clustering~\cite{chromatic_clustering}, and to bi-clustering (co-clustering)~\cite{biclustering}.

\section{Concluding remarks}\label{sec:conc}
This paper initiates the investigation into local correlation clustering, devising 
algorithms with sublinear time and query complexity. The tradeoff between the running time of our algorithms
and the quality of the solution found  is close to optimal. Moreover, our solutions are amenable to simple implementations and they can also be naturally adapted to the distributed and streaming settings in order to improve their latency or memory usage.

The notion of local clustering introduced in this paper opens an interesting line of work, which might lead to various contributions in more applied scenarios.
For instance, the ability of local algorithms to (among others) quickly estimate the cost of the best clustering
can provide a powerful a primitive for decision-making, around which to build new data analysis frameworks.
The streaming capabilities of the algorithms may also prove useful in clustering large-scale evolving graphs: this might be applied to detect communities in on-line social networks.

Another intriguing question is whether one can devise other graph-querying models that allow for improved theoretical results  while being reasonable from a practical viewpoint.
The $O(n^{3/2})$-time constant-factor approximation algorithm using neighborhood oracles that we discussed in Section~\ref{sec:extensions} suggests that
this may be a fruitful direction to pursue in further research. The question seems particulary relevant in order to apply local techniques to very sparse
graphs.



\newpage
\appendix
\section{A sharper bound for Lemma 4.10}\label{sec:sharper}

\begin{lemma}\label{lem:a_i}
Let $a_0 \in (0, 1)$ and define a sequence by $a_{i+1} = a_i (1 - a_i)$ for $i \ge 1$.
Then for all $j \ge 1$,
$$ a_j \le \frac{1}{\frac{1}{\widehat{a}_0} + j + \ln{(\widehat{a}_0 j)}-o_j(1)}, $$ 
where  $\widehat{a}_0 = \min(a_0, 1-a_0)$.
\end{lemma}
\begin{proof}
Since replacing $a_0$ with $1-a_0$ does not affect the terms $a_j$ for $j \ge 1$, we can assume
$a_0 = \widehat{a}_0 \le 1/2$, in which case the result holds also for $j = 0$. Set
$m_i = \frac{1}{a_i}$ for all $i \ge 0$. Then $m_0 = \frac{1}{a_0} \ge 2$ and
$$m_{i+1}-m_i = \frac{1}{\frac{1}{m_i} (1-\frac{1}{m_i})}-m_i = \frac{m_i^2}{m_i - 1}-m_i =  1 +
\frac{1}{m_i -
1}.$$
In other words, for all $i \ge 1$ we have
$$ m_i = m_0 + i + \sum_{j=0}^{i-1} \frac{1}{m_j - 1} .$$

Thus $m_i \ge m_1 \ge 4$, $m_i \le m_0 + \frac{4}{3} i\le m_0 + 2i+1$, and
          $$ m_i \ge m_0 + i + \sum_{j=1}^{i-1} \frac{1}{m_0 + 2 j} =
                     m_0 + i + \frac{1}{2} \sum_{j=1}^{i-1} \frac{1}{\frac{m_0}{2} + j}. $$
Since by the integral test
$$ \ln \left( \frac{b+1}{a} \right) \le \sum_{k=a}^b \frac{1}{k} \le  \ln \left( \frac{b}{a-1}
        \right), $$
we have
$$ \sum_{j=1}^{i-1} \frac{1}{\frac{m_0}{2} + j} \ge \ln\left(1+\frac{2 (i-1)}{m_0 + 2} \right) \ge
\ln\left( 1+\frac{i-1}{m_0} \right), $$
so
$$ m_i \ge m_0 + i + \frac{1}{2} \ln\left(1+ \widehat{a}_0 \cdot (i - 1) \right),$$
 as we wished to show.
\end{proof}

\section{Finding cut decompositions implicitly}\label{app:fk}
We give here an overview of Frieze and Kannan's method~\cite{approx_matrices}.
In essence, the process works with the submatrix induced by certain randomly chosen subsets $U, V$ of size $\poly(1/\eps)$
and defining $\indic[x \in R_i]$ and $\indic[y \in R_i]$ in terms of the adjacencies (matrix entries) of $x$ and $y$ with $U$ and $V$.

We start with the following simple exponential-time algorithm for finding cut decompositions. Suppose we have
found cut matrices $D_0, \ldots, D_{i-1}$ and we want to find $D_i$. Let $W_i = A - \sum_{j < i} D_j$
be the residual matrix. While there exist sets $R_i', S_i'$ with
\begin{equation}\label{eq:weight}
|W_i(R_i', S_i')| \ge \eps \sqrt{|R_i'| |S_i'|} \sqrt{m n},
\end{equation}
let $R_i = R_i', S_i = S_i'$, $d_i= W_i(R_i, S_i)/(|R_i| |S_i|)$
and add $D_i = CUT(R_i, S_i, d_i)$ to the decomposition. An easy computation shows that the squared Frobenius norm of the residual matrix decreases
by $W_i(R_i', S_i')^2 / (|R_i'| |S_i'|)$, i.e., at least an $\eps^2$ fraction of $\norm{A}_F^2 \le mn$. Therefore this process cannot go on for more than $1/\eps^2$ steps. This gives a non-constructive proof of existence of cut decompositions.

How to make this procedure run in time independent of the matrix size?
We can cut some slack here by replacing $\eps$ with some polynomial of $\epsilon$ with a larger
exponent. Frieze and Kannan pick a row set $R_i \subseteq \calR$ and then use a
sampling-based procedure to construct a column set $S_i \subseteq
\calS$ such that the $R_i\times S_i$
submatrix is sufficiently dense. Provided that the entries in the matrix $W_i$ remain bounded and inequality~\eqref{eq:weight} holds for some $R_i', S_i'$,
          they are able to find           $R_i\in\calR, S_i\in\calS$ such that
          $$|W_i(R_i, S_i)| \ge \poly(\eps) \cdot  m n,$$
which implies an $\poly(\eps)$-fractional decrease in the squared Frobenius norm of the residual
matrix.  They show that, with probability at least $\poly(\eps)$, we can take for $P_i$
the set of all $x\in \calR$ with $W_i(x, v) \cdot \nu \ge \nu^2$ for some randomly chosen $v \in \calC$ and $\nu \in[-1,1]$;
   and for $S_i$ the set of all $y \in \calC$ with $W_i(R_i, y) \cdot \nu \ge 0$. 

We need to deal with how to represent the sets $R_i, S_i$ used in the
decomposition in an implicit manner. We will write down a predicate that, given $i \in [s]$ and $x \in \calR$ (resp., $y \in \calS$), tells us
whether $x \in R_i$ (resp., $y \in S_i$) and can be evaluated in time $\poly(1/\eps)$ by making queries to $A$.
 Although the size of~$R_i\subseteq \calR$ may be linear in $m$, its definition    makes it possible to
 check for membership in $R_i$ with one query to $W_i$.   The set $S_i$, for its part, does not admit such a quick membership test, so Frieze and Kannan work
   with an approximation achieved by replacing $R_i$ with a $\poly(1/\eps)$-sized portion thereof in the definition of $S_i$.
   With the new definition, membership in $R_i$ and $S_i$ can be computed in time $\poly(1/\eps)$, as we shall see.
   Also, the density $d_i=W_i(R_i,S_i)/(|R_i| |S_i|)$ can be estimated to
   within $\pm \epsilon^2 mn / 16$ accuracy by sampling with $\poly(1/\eps)$ queries to $W_i$.

\mycomment{
            \newpage
\begin{algorithm}\label{alg:balls}
\begin{algorithmic}[0]

\Function{InR}{$i, x$}
\State \Return $\Call{W}{i, v_i, x}$.
\EndFunction
\Statex

\Function{InS}{$i, y$}
\State \Return $\sum{j < i} \sum_{v \in U_j} \Call{W}{i, v, x}$.
\EndFunction
\Statex

\Function{W}{$i, x, y$}
    \State $a \gets A_{x,y}$.       \Comment{Query the input matrix.}
    \For $j \in \{0,\ldots,i\}$
    \EndFor
\EndFunction
\Statex

\For {$i \in \poly(1/\eps)$}
    \State Pick $v_i \gets \calC$ at random.
    \State Pick $\nu_i \gets [-1,1]$ at random.
    \State Get a random sample $U_i$ of size $\poly(1/\eps)$.
    \State $U_i \gets \{ x \in \calR \mid W_i(v_i, x) \cdot \nu \ge \nu^2 \}$.
\EndFor
\end{algorithmic}
\end{algorithm}
}

Summarizing, we can build a cut decomposition in the following way. Let $s = \poly(1/\eps)$.
At at stage $i$, $i=0,1\ldots,s-1$, the first $i$ cut matrices $CUT(R_i,
        S_i, d_i)$ are implicitly known. Given the previous $i$ cut matrices, the residual matrix $W_i$ is given by
$$ W_i(x, y) = A(x,y) - \sum_{j < i} d_{j} \cdot \indic[x \in R_j] \cdot \indic[y \in S_j]; $$
extend the notation to sets in the obvious manner.

The set $R_i$ is defined in terms of a random element $v_i \in \calC$ and a random real $\nu_i \in [-1,1]$ by
$$ R_i = \{ x\in\calR \mid W_i(x, v_i) \cdot \nu \ge \nu^2 \}. $$
The set $S_i$ is defined in terms of $R_i$ and a random sample $U_i \subseteq \calR$ of size $\poly(1/\eps)$ by 
$$ S_i = \{ y\in\calS \mid W_i(U_i \cap R_i, y) \cdot \nu \ge 0 \}.$$
Finally, the density $d_i$ is defined in terms of another random sample $Z_i \subseteq \calR \times \calS$ of size $\poly(1/\eps)$ by
$$ d_i = \expect_{(x, y) \in Z_i}[W_i(x, y) \mid (x, y) \in R_i \times S_i] .$$

Let $U = \bigcup_i (U_i \cup \Pi_1(Z_i))$, $V = \bigcup_i (V_i \cup \Pi_2(Z_i))$. We need to compute compute $W_i(u, v)$, $\indic[u \in R_j]$,
    $\indic[v \in S_j]$ and $d_j$ for all $(u, v) \in U \times V$, $j \le i$.
    This can be done in time $\poly(s/\eps)$ using dynamic programming and the formulas above.
    This allows us to compute $W_i(x, y)$ for all $x, y$ in time $\poly(s/\eps) = \poly(1/\eps)$.

\end{document}